\documentclass[11pt]{article}
\usepackage{fullpage}

\usepackage{times}
\usepackage[margin=1in,footskip=0.3in]{geometry}
\usepackage{verbatim}
\usepackage[english]{babel}
\usepackage[utf8]{inputenc}

\usepackage{enumitem}
\usepackage{algorithm}
\usepackage[noend]{algpseudocode}
\usepackage{url}
\usepackage{color}

\usepackage{graphicx}
\usepackage[font=small]{caption}
\usepackage[font=small]{subcaption}
\usepackage{tikz}
\usetikzlibrary{calc}
\usetikzlibrary{shapes.geometric}
\usetikzlibrary{arrows}
\usetikzlibrary{decorations.pathreplacing, decorations.pathmorphing}
\usetikzlibrary{patterns}
\usetikzlibrary{backgrounds}
\usetikzlibrary{scopes}
\usetikzlibrary{arrows, automata, positioning}
\usetikzlibrary{decorations.markings}

\usepackage{amssymb,amsmath,amsthm,epsfig}
\usepackage{mathtools}
\usepackage{thm-restate}

\newtheorem{theorem}{Theorem}
\newtheorem{lemma}[theorem]{Lemma}

\newtheorem{invariant}{Invariant}

\newtheorem{definition}[theorem]{Definition}

\newcommand{\opt}{\mathrm{opt}}
\newcommand{\OPT}{\mathrm{OPT}}
\newcommand{\cred}{\mathrm{credits}}
\newcommand{\nc}{n_{\mathrm{comp}}}

\def\Iscr{\mathcal{I}}

\newcommand{\eps}{\varepsilon}
\renewcommand{\epsilon}{\varepsilon}

\usepackage{hyperref}

\usepackage[colorinlistoftodos,textsize=tiny,textwidth=2cm,color=red!25!white,obeyFinal]{todonotes}

\newcommand{\apxFAP}{1.9973}

\def\DEBUG{true}

\ifdefined\DEBUG

  \newcommand{\fabr}[1]{\todo[color=red!100!black!50]{F: #1}}

\else
  
  \newcommand{\fabr}[1]{}

\fi

\usepackage{authblk}

\begin{document}

\title{Breaching the 2-Approximation Barrier\\ for the Forest Augmentation Problem\thanks{The first author is partially supported by the SNF Excellence Grant 200020B\_182865/1 and by the SNF Grant 200021\_200731/1.
The third author is supported by the Swiss National Science Foundation grant 200021\_184622.
}}

\author[1]{Fabrizio Grandoni}
\author[2]{Afrouz Jabal Ameli}
\author[3]{Vera Traub}
\affil[1,2]{IDSIA, USI-SUPSI}
\affil[3]{ETH Zurich}

\date{}
\pagenumbering{gobble}
\maketitle
\begin{abstract}
\noindent The basic goal of survivable network design is to build cheap networks that guarantee the connectivity of certain pairs of nodes despite the failure of a few edges or nodes. A celebrated result by Jain [Combinatorica'01] provides a $2$-approximation for a wide class of these problems. However nothing better is known even for very basic special cases, raising the natural question whether any improved approximation factor is possible at all. 

In this paper we address one of the most basic problems in this family for which $2$ is still the best-known approximation factor, the Forest Augmentation Problem (FAP): given an undirected unweighted graph (that w.l.o.g. we can assume to be a forest) and a collection of extra edges (links), compute a minimum cardinality subset of links whose addition to the graph makes it $2$-edge-connected. Several better-than-$2$ approximation algorithms are known for the special case where the input graph is a tree, a.k.a. the Tree Augmentation Problem (TAP), see e.g. [Grandoni, Kalaitzis, Zenklusen - STOC'18; Cecchetto, Traub, Zenklusen - STOC'21] and references therein. 
Recently this was achieved also for the weighted version of TAP [Traub, Zenklusen - FOCS'21], and for the $k$-connectivity generalization of TAP [Byrka, Grandoni, Jabal-Ameli - STOC'20; Cecchetto, Traub, Zenklusen - STOC'21]. These results heavily exploit the fact that the input graph is connected, a condition that does not hold in FAP.

In this paper we breach the $2$-approximation barrier for FAP. Our result is based on two main ingredients. First, we describe a reduction to the Path Augmentation Problem (PAP), the special case of FAP where the input graph is a collection of disjoint paths. Our reduction is \emph{not} approximation preserving, however it is sufficiently accurate to improve on a factor $2$ approximation.
Second, we present a better-than-$2$ approximation algorithm for PAP, an open problem on its own. 
Here we exploit a novel notion of \emph{implicit credits} which might turn out to be helpful in future related work.

\end{abstract}

\newpage
\pagenumbering{arabic}

\setcounter{page}{1}

\section{Introduction}

Real networks are prone to failures. The basic goal of \emph{survivable network design} is to build ``cheap'' networks that provide connectivity between given pairs of nodes despite the failure of a few edges or nodes (we will next focus on edge faults). Most natural survivable network design problems are NP-hard, therefore it makes sense to study them in terms of approximation algorithms. A celebrated result by Jain \cite{J01} provides a $2$-approximation for a very wide family of such (edge-connectivity) problems. However, nothing better is known even for very basic special cases. Therefore a recent trend in the area is trying to breach the $2$ approximation barrier for interesting special cases of survivable network design. 
 
In this paper we focus on one of the most basic survivable network design problems for which the best-known approximation factor is still $2$, namely the Forest Augmentation Problem (FAP): given an undirected unweighted graph $(V,F)$ and a collection of extra edges $L\subseteq {{V}\choose{2}}$, called \emph{links}, find a minimum cardinality subset of links $S\subseteq L$ such that $(V,F\cup S)$ is $2$-edge-connected\footnote{We recall that a graph is $k$-edge connected if it remains connected even after removing an arbitrary subset of up to $k-1$ edges.}. 
Observe that one obtains an equivalent problem by contracting each $2$-edge-connected component of $(V,F)$ into a single node, hence leading to a forest graph (motivating the name FAP).
Therefore, we will assume in the following that $(V,F)$ is a forest.

FAP generalized two well-studied problems for which better-than-$2$ approximation algorithms are known, namely the Tree Augmentation Problem (TAP) \cite{A17,CTZ21,CN13,EFKN09,FGKS18,GKZ18,KKL04,KN16b,N03,N17,TZ21,TZ22} and the $2$-Edge-Connected Spanning Subgraph problem ($2$-ECSS) \cite{CSS01,CT00,GG12,HVV19,KV94,SV14}. Therefore a natural question is whether also FAP admits an approximation factor strictly below $2$. Indeed, this is explicitly posed as an open problem, e.g., in \cite{CTZ21,CCDZ20,CDGKN20}.
This question was open even in the special case of FAP where the input forest is a collection of disjoint paths, also known as the Path Augmentation Problem (PAP).

TAP is the special case of FAP where the forest consists of a single spanning tree. $2$-ECSS is the problem of computing a $2$-edge-connected subgraph of an input graph $G=(V,E)$ with the minimum number of edges. Hence $2$-ECSS can be interpreted as the special case of FAP where all trees in the forest are singleton nodes (i.e. $F=\emptyset$). One can define a natural weighted generalization $2$-WECSS of $2$-ECSS with edge weights, where the goal is to compute a $2$-edge-connected subgraph of minimum total weight. Then FAP is the special case of $2$-WECSS with edge weights $0$ (for $e\in F$) and $1$ (for $e\in L$). Therefore an improved approximation for FAP is a first step in the direction of a similar result for $2$-WECSS, a well-known open problem in the area.

In this paper we breach the $2$-approximation barrier for FAP, namely we obtain the following result.
\begin{theorem}\label{thr:main}
There is a polynomial-time deterministic $\apxFAP$-approximation algorithm for FAP. 
\end{theorem}
In Section \ref{sec:overview} we provide an overview of the main ideas behind the above result.

\subsection{Previous and Related Work}

TAP (hence FAP) is known to be APX-hard \cite{KKL04}, thus in particular it does not admit a PTAS unless $P=NP$. Several better-than-$2$ approximation algorithms are known for this problem \cite{EFKN09,GKZ18,KN16b,N03}, culminating with a $1.393$ approximation by Cecchetto, Traub and Zenklusen \cite{CTZ21}. 

Very recently the $2$-approximation barrier was breached for the natural weighted version WTAP of TAP (with link weights) by Traub and Zenklusen \cite{TZ21}, who presented a $1.694$ approximation (which they improved to a $(1.5+\epsilon)$-approximation in \cite{TZ22}). This solved one of the main open (and simplest weighted) problems in the area (preliminary results in this directions, among others, appeared in \cite{A17,CN13,FGKS18,GKZ18,N17}).

FAP and PAP belong to the family of \emph{connectivity augmentation} problems, where the goal is to increase the connectivity of an existing network by adding links. One important problem in this area is the $k$-Connectivity Augmentation Problem ($k$-CAP): given a $k$-edge-connected graph $G=(V,E)$ and a collection of links $L\subseteq {{V}\choose{2}}$, find a minimum cardinality subset of links $S\subseteq L$ such that $G'=(V,E\cup S)$ is $(k+1)$-edge-connected. Hence in particular TAP is the special case of $k$-CAP with $k=1$. Known approximation-preserving reductions \cite{DKL76} show that $k$-CAP reduces to the case $k=2$, a.k.a. the Cactus Augmentation Problem (CacAP). After some preliminary results on special cases \cite{GGJS19}, a better-than-$2$ approximation for CacAP was achieved recently by Byrka, Grandoni and Jabal-Ameli \cite{BGJ20} via a (non-black-box) reduction to the Steiner Tree problem \cite{BGRS13} (see also \cite{N20} for a black-box reduction to the same problem). This was improved to $1.393$ by Cecchetto et al. \cite{CTZ21} with a completely different approach, more similar in spirit to prior work on TAP.

Better-than-$2$ approximation algorithm are known for the Matching Augmentation Problem (MAP),i.e., the special case of FAP where the input forest is a matching \cite{CCDZ20,CDGKN20}. In more detail, Cheriyan, Dippel, Grandoni, Khan and Narayan \cite{CDGKN20} present a $7/4$ approximation for MAP. 
This was later improved to $5/3$ by Cheriyan, Cummings, Dippel and Zhu \cite{CCDZ20}. Finding a better-than-$2$ approximation for FAP is mentioned in \cite{CCDZ20,CDGKN20} as one of the main motivations to study MAP. 
We remark that this question was open even for the Path Augmentation Problem (PAP). The techniques in \cite{CCDZ20,CDGKN20} do not seem to extend even to paths of length $2$.

Khuller and Vishkin \cite{KV94} found the first better-than-$2$ (namely a $3/2$) approximation for $2$-ECSS. Cheriyan, Seb\H{o} and Szigeti \cite{CSS01} improved the approximation factor to $17/12$. The current best known approximation factor for $2$-ECSS is $4/3$, and this can be achieved with two rather different approaches \cite{HVV19,SV14}. Hunkenschr{\"{o}}der, Vempala and Vetta \cite{HVV19} use a credit invariant (vaguely) similar to the one used in this paper.
Seb{\H{o}} and Vygen \cite{SV14} instead exploit an ear decomposition of the input graph with special properties. 
The natural generalization $k$-ECSS of $2$-ECSS to $k$ connectivity was studied, among others, by \cite{CT00,GG12}. 

\section{Preliminaries}\label{sec:preliminaries}

In this section we first introduce some notation (Section~\ref{sec:notation}) and then recap some algorithms from previous work that we will build on.
Specifically, in Section~\ref{sec:wtap} we explain the results on WTAP from \cite{TZ21} that we will use in our algorithm, while in Section~\ref{sec:w2ecss} we describe a well-known 2-approximation algorithm for W2ECSS (and thus also for FAP).
In what follows all algorithms will be deterministic and we will therefore not mention this explicitly anymore.

\subsection{Notation}\label{sec:notation}

We use standard graph notation. In particular, given a graph $G=(V,E)$ and $\emptyset\neq R\subsetneq V$, by $\delta(R)\subseteq E$ we denote the edges with exactly one endpoint in $R$. If $G$ is directed, then $\delta^{-}(R)\subseteq E$ are the edges (or arcs) entering $R$. For any $F\subseteq E$, we let $\delta_F(R)=\delta(R)\cap F$ and $\delta^{-}_F(R)=\delta^{-}(R)\cap F$. 
For an edge weight function $w$ and $F\subseteq E$, $w(F):=\sum_{e\in F}w(e)$.

Let $(V,F,L)$ be a considered instance of FAP. W.l.o.g. we assume that $(V,F\cup L)$ is $2$-edge-connected (otherwise there is no feasible solution). By $\nc:=|V|-|F|$ we denote the number of connected components (or components for short) of the forest (in particular in a TAP instance $\nc=1$). By $\OPT\subseteq L$ we denote an optimal solution to this instance and by $\opt=|\OPT|$ its size. Notice that $\opt\geq \nc$.

We will call the elements of $L$ \emph{links}.
By \emph{edges} we will mean both links and edges of the input forest.

\subsection{Preliminaries on WTAP}\label{sec:wtap}

Given a WTAP instance $(V,F,L,w)$ and $\ell=\{a,b\}\in L$, we let $P_{\ell}$ be the (edge set of the) path in the tree $(V,F)$ between endpoints $a$ and $b$. We say that $\ell$ covers  the edges of $P_{\ell}$. Then a feasible solution to WTAP is a subset of links that covers all the edges. We say that $\ell'\in L$ is a shadow of $\ell\in L$ if $P_{\ell'}\subseteq P_{\ell}$. W.l.o.g. we can assume that all the possible shadows of a link $\ell$ are present in the input. Indeed, if this is not the case, we can add any missing shadow of $\ell$ with weight $w(\ell)$: any feasible solution for the new problem can be converted into a feasible solution for the original problem and vice versa. For similar reasons, if $\ell'$ is a shadow of $\ell$, we can assume that $w(\ell')\leq w(\ell)$.

Our algorithm for FAP will make use of recent insights on WTAP from \cite{TZ21}.
Thus, we briefly summarize the statements from \cite{TZ21} which we will need.
Consider a WTAP instance $(V,F,L,w)$, that we interpret as rooted at some node $r$. We define an \emph{up-link} as a link where one endpoint in an ancestor of the other endpoint in the tree $(V,F)$. 
The approximation algorithm for WTAP from \cite{TZ21}, a so-called relative greedy algorithm, first computes a 2-approximation $U\subseteq L$ that contains only up-links; it is well-known that this exists and is computable in polynomial time.
Then the main insight from \cite{TZ21} is the following lemma that allows for improving this $2$-approximate solution $U$ to a $(1+\ln(2)+\epsilon)$-approximation $\mathrm{APX}$ (for any constant $\epsilon > 0$).

\begin{lemma}[\cite{TZ21}]\label{lem:wtap} 
\footnote{This lemma is not stated explicitly in this form in \cite{TZ21} and thus we provide a detailed explanation of how this statement follows from \cite{TZ21} in the appendix.}
Let $(V,F,L,w)$ be a WTAP instance rooted at some node $r$ with optimal solution $\OPT$, and let $U\subseteq L$ be a feasible solution consisting only of up-links. 
For any $\epsilon > 0$, one can in polynomial time  construct a feasible WTAP solution $\mathrm{APX}$ with weight
\[
w(\mathrm{APX}) \leq \left(1+\ln\left(\frac{w(U)}{w(\OPT)}\right)+\epsilon\right)\cdot w(\OPT).
\] 
\end{lemma}

We will apply the above lemma to up-link solutions obtained in a different way.
While in general there might not exist a better-than-2 approximation for WTAP that consists only of up-links, in some cases we can obtain such very good up-link solutions, as we will explain in Section~\ref{sec:overview}.
In particular, we highlight that we do not use the $(1+\ln(2)+\epsilon)$-approximation algorithm for WTAP as a black-box, but crucially use the properties of the relative greedy algorithm that allow for improving arbitrary up-link solutions, as stated in Lemma~\ref{lem:wtap}.
For this reason, we do not use the recent improvements on \cite{TZ21} that lead to a $(1.5+\epsilon)$-approximation algorithm for WTAP \cite{TZ22}.

\subsection{A Simple 2-Approximation for 2-WECSS}\label{sec:w2ecss}
We next sketch how to obtain a $2$-approximation for a $2$-WECSS instance $(V,E,w)$  because we will need this algorithm from \cite{KV94} later. 
Let us replace each undirected edge $e=\{a,b\}$ with two oppositely directed arcs $(a,b)$ and $(b,a)$ with weight $w(e)$. 
We call $A$ this set of directed arcs and still use $w(\cdot)$ to denote their weight. 
Moreover, we fix an arbitrary vertex $r$ as the root. 
If we can find an arc set $D\subseteq A$ such that every set $\emptyset \neq R \subseteq V\setminus \{r\}$ has two entering arcs, then clearly the set of undirected edges corresponding to $D$ is a feasible solution to our $2$-WECSS instance.
We remark, that by Edmonds' disjoint branching theorem (see, e.g., Corollary~53.1c in \cite{Schrijver}), such arc sets $D$ are precisely those that contain two edge-disjoint spanning arborescences rooted at $r$.

Moreover, the cheapest such set $D$ satisfies the following claim:
\begin{lemma}\label{lem:2apx}
Given a $2$-WECSS instance $(V,E,w)$ with optimal weight $\opt$ and the corresponding directed arc set $A$,
we have 
$$
\min \Big\{ w(D) : D\subseteq A,\ |D\cap \delta^-(R)| \ge 2 \text{ for all } \emptyset \ne R \subseteq V\setminus \{r\} \Big\} \ \le 2 \cdot \opt.
$$
\end{lemma}
\begin{proof}
Indeed, by taking an optimal $2$-WECSS solution and replacing every edge by the two corresponding directed arcs we obtain a feasible set $D\subseteq A$.
\end{proof}
A $2$-approximation for $2$-WECSS (and hence for FAP) is implied from the above discussion and the following lemma.
\begin{lemma}\label{lem:matroid_intersection}
Given a directed graph $(V,A)$ with nonnegative arc-weights $w : A \to \mathbb{R}_{\ge 0}$ and a root vertex $r\in V$ we can find in polynomial time a minimum-weight set $D\subseteq A$ such that 
\[
 |D\cap \delta^-(R)|\ \ge\ 2 \quad \text{ for all } \emptyset \ne R \subseteq V\setminus \{r\}.
\]
\end{lemma}
\begin{proof}
By Edmonds' disjoint branching theorem (Corollary~53.1c in \cite{Schrijver}), the desired sets $D$ are precisely those sets that contain two disjoint spanning arborescences rooted at $r$.
Hence, because the edge weights are non-negative, we can find a minimum-weight set $D\subseteq A$ as desired by computing a cheapest edge set that is the union of two disjoint arborescences. 
This can be done in polynomial time; see Theorem~53.10 in \cite{Schrijver}.
\end{proof}

\section{Overview of Our Approach}\label{sec:overview}

Our main result (Theorem \ref{thr:main}) follows by combining two different algorithms: the first one (Section~\ref{sec:few_components}) provides a good approximation for the case that $\opt$ is much larger than $\nc$, and the second one (Section \ref{sec:many_components}) provides a good approximation for the complementary case. In particular, we prove the following. 
\begin{restatable}{lemma}{tapalgo}\label{lem:result_tap_algo}
Let $\epsilon > 0$ be a constant.
Given an instance $(V,F,L)$ of FAP, we can compute in polynomial time a solution of size at most
\[
\nc + \left( 1 + \ln\Big(2-\frac{\nc}{\opt}\Big) + \epsilon\right) \cdot \opt.
\]
\end{restatable}

\begin{lemma}\label{lem:result_credit_algo}
Given an instance $(V,F,L)$ of FAP, we can compute in polynomial time a solution of size at most
\[
\tfrac{7}{4} \cdot \opt+ \tfrac{13}{4} \cdot \left(\opt -\nc\right).
\]
\end{lemma}
Lemma~\ref{lem:result_tap_algo} yields an approximation ratio ranging from $2$ to $1+\ln 2 + \epsilon < 1.7$ as the ratio $\frac{\nc}{\opt}$ decreases from $1$ to $0$\footnote{For $\nc=1$, our algorithm is indeed identical to the one in \cite{TZ21}.}. Hence, to improve on a $2$ approximation, it is sufficient to have such an approximation for the case that $\frac{\nc}{\opt}$ is close to $1$: this is achieved by the algorithm from Lemma \ref{lem:result_credit_algo}.  Theorem \ref{thr:main} follows easily.
\begin{proof}[Proof of Theorem \ref{thr:main}]
Consider the best solution among the ones returned by the algorithms from Lemmas \ref{lem:result_tap_algo} and \ref{lem:result_credit_algo}. Let $\alpha\in [0,1]$ be such that $\nc=\alpha \cdot \opt$. Then the approximation factor of this algorithm is at most $\min\{\alpha+1+\ln(2-\alpha),\ \tfrac{7}{4}+\tfrac{13}{4}(1-\alpha)\}+\eps$. The worst-case ratio is obtained by choosing $\alpha$ so that the two terms in the minimum are equal, implying the claim. 
(This is the case for $\alpha\approx 0.923925$.)
\end{proof}
We next sketch the basic ideas behind the above two lemmas. 

\subsection{Overview of the Algorithm for a Small Number of Connected Components}

First, we explain how we obtain Lemma~\ref{lem:result_tap_algo}, which yields a good approximation ratio if the  number $\nc$ of connected components of the forest is much smaller than $\opt$.
In this case, it seems natural to consider the simple reduction to TAP: complete the forest $(V,F)$ to a spanning tree by adding $\nc - 1$ links and then apply an approximation algorithm for TAP to obtain a 2-edge-connected graph.
Using a $\rho$-approximation algorithm for TAP this yields a FAP solution of size at most $\nc - 1 + \rho \cdot \opt$. At the moment, the best known approximation algorithm allows for choosing $\rho=1.393$ \cite{CTZ21}.
While this yields a better-than-2 approximation if $\nc$ is sufficiently smaller than $\opt$, i.e., if $\nc \le (2-\rho-\delta) \opt$ for some constant $\delta>0$, this guarantee is not good enough for our purposes.
In contrast, the statement of Lemma~\ref{lem:result_tap_algo} is much stronger in the sense that it yields an improvement over the approximaton ratio of $2$ as soon as $\nc \le (1-\delta) \opt$ for any arbitrary constant $\delta >0$.

Thus, let us consider the following, slightly more involved reduction to TAP.
Let $S\subseteq L$ be a 2-approximate solution for FAP.
Then we can complete the forest $(V,F)$ to a spanning tree $(V,T)$ using only edges from $S$.
This way, we obtain a tree $(V,T)$ and a solution $S\setminus T$ for the TAP problem of augmenting $(V,T)$ to a 2-edge-connected graph.
We observe that $|T\cap S|=\nc - 1$ and $|S\setminus T| = |S| - \nc + 1 \le 2 \cdot \opt - \nc + 1$.
In particular $S\setminus T$ can only be an optimal solution to the TAP instance if $\nc \approx \opt$ and $S$ is not already a better-than-2 approximation (in which case we are done).
Because $S$ is a 2-approximation for our FAP instance, we would immediately obtain a better-than-2 approximation if we could improve the TAP solution $S\setminus T$ by just a little bit, similar as it is done in Lemma~\ref{lem:wtap}.
Unfortunately, we cannot apply Lemma~\ref{lem:wtap} because $S\setminus T$ will in general consist not only of up-links.

In order to address this issue, we use the 2-approximation algorithm from Section~\ref{sec:w2ecss} not as a black-box, but exploit that the computed solution $S$ has stronger properties than just being a FAP solution that is not too expensive.
More precisely, we show that from the directed arc set $D$ computed in the algorithm from Section~\ref{sec:w2ecss} we can obtain a spanning tree $(V,T)$ with $F\subseteq T$ and a TAP solution consisting of at most $|S| - \nc + 1$ many up-links.
Then, applying Lemma~\ref{lem:wtap}, completes the proof of Lemma~\ref{lem:result_tap_algo}.
In order to obtain the up-link solution from the directed arc set $D$, we use that directed links can in a certain sense be interpreted as up-links, an observation first made in \cite{CTZ21}.
For details on how we prove Lemma~\ref{lem:result_tap_algo}, we refer to Section~\ref{sec:few_components}.

\subsection{Overview of the Algorithm for a Large Number of Connected Components}
In view of Lemma~\ref{lem:result_tap_algo}, it makes sense to design approximation algorithms for the case that $\frac{\nc}{\opt}$ is close to~$1$. This motivates the following definition. 
\begin{definition}
An algorithm for FAP is a $(\rho,K)$-approximation algorithm if it produces a solution of size at most \[ \rho\cdot \opt+K \cdot (\opt-\nc). \]
\end{definition}
If we can find a $(\rho,K)$-approximation algorithm for some $\rho < 2$ and any arbitrary constant $K\ge 0$, then together with Lemma~\ref{lem:result_tap_algo}, this implies a better-than-2 approximation for FAP.

We will show how to obtain a $(\tfrac{7}{4}, \tfrac{13}{4})$-approximation algorithm for FAP, which is precisely the statement of Lemma~\ref{lem:result_credit_algo}.
To achieve this we exploit a relatively simple reduction to PAP. 
In order to simplify the notation, we can further impose that the PAP instance has no isolated nodes, i.e., all connected components of the forest $(V,F)$ are paths with length at least $1$.

\begin{lemma}\label{lem:reducing_to_pap_noIsolated}
Given a polynomial time $(\rho,K)$-approximation algorithm for PAP without isolated nodes for some constants $\rho \ge 1$ and $K>0$, there is polynomial-time $(\rho, K + 2(\rho-1))$-approximation algorithm for FAP.
\end{lemma}

We prove Lemma~\ref{lem:reducing_to_pap_noIsolated} in the appendix (Section~\ref{sec:reductionPAP}).
We conclude that in order to prove Lemma~\ref{lem:result_credit_algo}, it suffices to show the following.

\begin{lemma}\label{lem:result_pap_algo}
There is a $(\tfrac{7}{4}, \tfrac{7}{4})$-approximation algorithm for PAP without isolated nodes.
\end{lemma}

In the rest of this section we sketch the proof of Lemma \ref{lem:result_pap_algo}. The basic algorithm is analogous to known algorithms for MAP \cite{CCDZ20,CDGKN20} and 2-ECSS \cite{HVV19}. 
We start by building (in polynomial time) an infeasible partial solution $S$. 
Then we gradually modify $S$ by adding (and sometimes removing) links until we obtain a feasible solution.

In \cite{CCDZ20,CDGKN20,HVV19} the initial solution is a 2-edge-cover\footnote{We recall that a $2$-edge cover of a graph is a subgraph where each node has degree at least $2$.} with the minimum number of links. 
The number of links in such a 2-edge-cover provides a lower bound on $\opt$.
A cheapest 2-edge cover can be obtained by first computing a maximum cardinality matching (in $L$) on the leaves of the forest $(V,F)$ and then adding an arbitrary incident link for every unmatched leaf.
Instead of using a cheapest 2-edge-cover as a starting solution, we use such a maximum matching $M\subseteq L$ between the leaves of the input paths, but we exclude links that match the two endpoints of a single path in $(V,F)$ (\emph{bad links}).
Intuitively, using such links is bad because they do not help to connect different connected components of the forest.
For this reason, also the optimal solution can use these bad links only if $\opt > \nc$, which we exploit to show that if many leaves remain unmatched, $\opt$ must be much larger than $\nc$.
(See Lemma~\ref{lem:matching_bound} for the precise statement.)
Working with the matching $M$ rather than a 2-edge-cover is useful because it allows us to exclude bad links (which we could not do in a 2-edge-cover), but also because it simplifies the later parts of the proof.\footnote{Alternatively, we could have worked with a weighted version of the 2-edge-cover where we make bad links more expensive, but working with the matching is overall simpler.}

In order to upper bound the size of the final solution we use a \emph{credit assignment scheme} similarly to \cite{CCDZ20,CDGKN20,HVV19}. The basic idea is to assign (nonnegative fractional) credits to certain parts of the current graph $H=(V,F\cup S)$ (like links, nodes, components, $2$-edge connected components etc.).
 Let $\cred(H)$ be the total number of credits assigned to $H$. We show that the initial graph $H=(V,F\cup S)$ with $S=M$, and every intermediate graph $H=(V,F\cup S)$ satisfies the following invariant:
\begin{invariant}\label{invariant:credits}
\begin{equation*}
\cred(H) + |S| \ \le\ \tfrac{7}{4}\opt +  \tfrac{7}{4} \left(\opt -\nc\right).
\end{equation*}
\end{invariant}
At the end of our algorithm, $S$ is a feasible solution and thus
Lemma \ref{lem:result_pap_algo} follows immediately because $\cred(H)\geq 0$. 

Our credit assignment scheme however critically deviates from prior work in the following sense: typically a credit invariant is \emph{explicit} meaning that, given $H$, one can compute in polynomial time the credits assigned to each part of $H$. We rather use an \emph{implicit} credit invariant where part of the credits are assigned based on properties of the (unknown) optimal solution $\OPT$. Our algorithm is able to work despite the lack of knowledge about the precise number of credits available. In prior related work the role of implicit credits is played by complicated preprocessing steps and more complex credit invariants and case analysis. We believe that implicit credits can be used to simplify and/or strengthen related results in the literature, and might be useful to address generalizations and variants of FAP in future work. 

Similar to prior work on MAP \cite{CCDZ20,CDGKN20},
our algorithm consists of two phases, called \emph{bridge-covering} and \emph{gluing}.
At the end of the bridge-covering phase, every connected component of $H=(V,F\cup S)$ will be 2-edge-connected.
In other words, $H$ is bridgeless, i.e., none of its edges is a bridge\footnote{A bridge is an edge whose removal would increase the number of connected components.}.
During the gluing phase, we maintain the property that $H$ is bridgeless and at the end of the gluing phase $H$ will be 2-edge-connected.
While this high-level structure is similar to prior work, we also introduce new concepts in this part, making the algorithm simpler and avoiding some issues that were in prior work addressed by rather complicated preprocessing steps.

For the bridge covering step we introduce the new concept of \emph{alternating trails}, which we define in Section~\ref{sec:bridge_covering}.
Using this concept, our algorithm for bridge covering can be simply stated as repeatedly doing the following until $H$ is bridgeless.
Consider the graph resulting from $H$ by contracting each 2-edge-connected component.
This graph is a forest and its edges are precisely the bridges of $H$.
Pick an arbitrary leaf of this forest, find an alternating trail that starts at this leaf and `covers as many bridges as possible'. 
Then augment $H$ along this trail.
For a precise description of the algorithm and definitions of \emph{alternating trails} and \emph{augmentation} we refer to Section~\ref{sec:bridge_covering}.
Note that it is crucial that when augmenting along alternating trails we do not only add links to $S$ in the bridge covering step, but also remove some links. 
Indeed, this is necessary to obtain a better-than-2 approximation even when $\tfrac{\nc}{\opt}$ is arbitrarily close to $1$ as shown in Figure~\ref{fig:removing_needed_bridge_covering}.

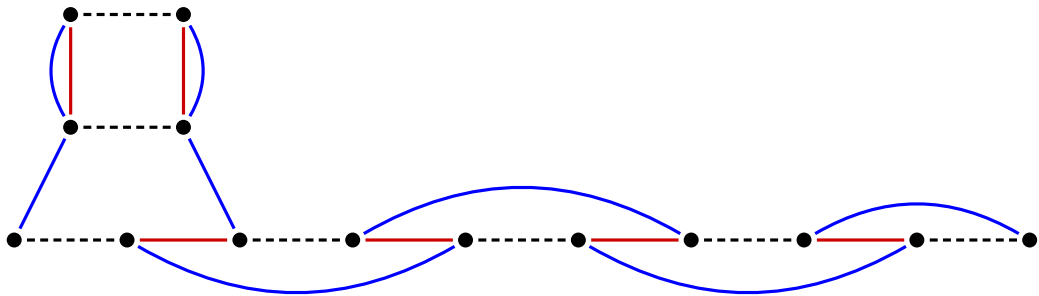
\begin{figure}[h]
\begin{center}
\begin{tikzpicture}[scale=1.5]

\tikzset{
vertex/.style={fill=black,circle,minimum size=0pt, inner sep=2pt, outer sep=2pt}
}
\tikzset{
block/.style={draw=black, very thick, circle,minimum size=0pt, inner sep=2pt, outer sep=2pt}
}

\begin{scope}[every node/.style={vertex}]
\node (1) at (1,1) {};
\node (2) at (2,1) {};
\node (3) at (3,1) {};
\node (4) at (4,1) {};
\node (5) at (5,1) {};
\node (6) at (6,1) {};
\node (7) at (7,1) {};
\node (8) at (8,1) {};
\node (9) at (9,1) {};
\node (10) at (10,1) {};

\node (a) at (1.5,2) {};
\node (b) at (2.5,2) {};
\node (c) at (1.5,3) {};
\node (d) at (2.5,3) {};
\end{scope}

\begin{scope}[blue, very thick]
\draw (1) -- (a);
\draw (b) -- (3);
\draw[bend right] (2) to (5);
\draw[bend right] (6) to (9);
\draw [bend left] (4) to (7);
\draw [bend left] (8) to (10);
\draw[bend left] (a) to (c);
\draw[bend left] (d) to (b);
\end{scope}

\begin{scope}[red!80!black, very thick]
\draw (2) -- (3);
\draw (4) -- (5);
\draw (6) -- (7);
\draw (8) -- (9);
\draw (a) -- (c);
\draw (b) -- (d);
\end{scope}

\begin{scope}[very thick, densely dashed]
\draw (1) -- (2);
\draw (3) -- (4);
\draw (5) -- (6);
\draw (7) -- (8);
\draw (9) -- (10);
\draw (a) -- (b);
\draw (c) -- (d);
\end{scope}

\end{tikzpicture}
\end{center}
\caption{
Example showing that even for $\nc\approx\opt$ we cannot ensure $|S|\le \rho \cdot \opt$ in the bridge covering step for any $\rho < 2$ if we never remove any links from $S$.
Solid edges are links from $L$, dashed edges are edges of the forest $F$.
The red links are a possible matching $M$ chosen in the initialization step of our algorithm.
The blue links are an optimal solution $\OPT$.
If the red matching $M$ is chosen, we need to add all but two links from $\OPT$ to $S$ in order to obtain a graph $(V,F\cup S)$ in which all connected components are 2-edge-connected.
For a large enough number of vertices, i.e., making the path to the right long enough, this shows that we cannot obtain a better-than-2 approximation without removing some links in the bridge covering step.
\label{fig:removing_needed_bridge_covering}
}
\end{figure}

The final part of our algorithm is the gluing step in which we maintain a bridgeless graph $H$ and transform it into a 2-edge-connected graph.
Our algorithm will repeatedly find a cycle $Q$ in the graph resulting from $(V,F\cup L)$ by contracting each 2-edge-connected component of $H$ and add the links from $Q$ to our current partial solution $S$ (and thus to $H=(V, F\cup S))$.
However, by adding only links and never removing any link in the gluing step it is impossible to obtain a better-than-2 approximation, even for $\opt=\nc$, as shown in Figure~\ref{fig:removing_needed_gluing}.
Therefore, in some cases where we need this to achieve the desired approximation ratio, our algorithm will exploit what we call a \emph{good cycle} (defined in Section~\ref{sec:gluing}) that allows us to add some links to $S$, but also to remove some links from $S$ which become redundant.

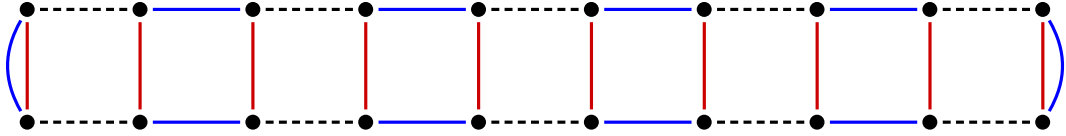
\begin{figure}[h]
\begin{center}
\begin{tikzpicture}[scale=1.5]

\tikzset{
vertex/.style={fill=black,circle,minimum size=0pt, inner sep=2pt, outer sep=2pt}
}
\tikzset{
block/.style={draw=black, very thick, circle,minimum size=0pt, inner sep=2pt, outer sep=2pt}
}

\begin{scope}[every node/.style={vertex}]
\node (1a) at (1,1) {};
\node (2a) at (2,1) {};
\node (3a) at (3,1) {};
\node (4a) at (4,1) {};
\node (5a) at (5,1) {};
\node (6a) at (6,1) {};
\node (7a) at (7,1) {};
\node (8a) at (8,1) {};
\node (9a) at (9,1) {};
\node (10a) at (10,1) {};

\node (1b) at (1,0) {};
\node (2b) at (2,0) {};
\node (3b) at (3,0) {};
\node (4b) at (4,0) {};
\node (5b) at (5,0) {};
\node (6b) at (6,0) {};
\node (7b) at (7,0) {};
\node (8b) at (8,0) {};
\node (9b) at (9,0) {};
\node (10b) at (10,0) {};
\end{scope}

\begin{scope}[blue, very thick]
\draw (2a) -- (3a);
\draw (4a) -- (5a);
\draw (6a) -- (7a);
\draw (8a) -- (9a);

\draw (2b) -- (3b);
\draw (4b) -- (5b);
\draw (6b) -- (7b);
\draw (8b) -- (9b);

\draw[bend left] (10a) to (10b);
\draw[bend right] (1a) to (1b);
\end{scope}

\begin{scope}[red!80!black, very thick]
\draw (1a) -- (1b);
\draw (2a) -- (2b);
\draw (3a) -- (3b);
\draw (4a) -- (4b);
\draw (5a) -- (5b);
\draw (6a) -- (6b);
\draw (7a) -- (7b);
\draw (8a) -- (8b);
\draw (9a) -- (9b);
\draw (10a) -- (10b);
\end{scope}

\begin{scope}[very thick, densely dashed]
\draw (1a) -- (2a);
\draw (3a) -- (4a);
\draw (5a) -- (6a);
\draw (7a) -- (8a);
\draw (9a) -- (10a);

\draw (1b) -- (2b);
\draw (3b) -- (4b);
\draw (5b) -- (6b);
\draw (7b) -- (8b);
\draw (9b) -- (10b);
\end{scope}

\end{tikzpicture}
\end{center}
\caption{
Example showing that we cannot obtain a better-than-2 approximation if we never remove any links in the gluing step.
Solid edges are links from $L$, dashed edges are edges of the forest $F$.
The red links are a possible matching $M$ chosen in the initialization step of our algorithm.
If this matching is chosen, all connected components of $(V,F\cup M)$ are 2-edge-connected and hence the bridge covering step does nothing.
In order to obtain any feasible solution, we need to include all but two of the edges from the optimal solution, which is shown in blue.
For a large enough number of vertices, this shows that we cannot obtain a better-than-2 approximation without removing some links in the gluing step and this holds even if $\nc=\opt$.
\label{fig:removing_needed_gluing}
.}
\end{figure}

Finally, we remark that we designed our approximation algorithm for PAP favouring simplicity over a (slightly) better approximation factor.

\section{An Algorithm for Forests with Few Connected Components}\label{sec:few_components}

In this section we prove Lemma~\ref{lem:result_tap_algo}, making use of Lemma~\ref{lem:wtap}.
First of all, we consider the $2$-WECSS instance $(V,F \cup L, w)$ corresponding to the input instance of FAP and we fix an arbitrary root $r\in V$.
Moreover, as in Section~\ref{sec:w2ecss}, we again denote by $A$ the arc set that, for every (undirected) edge in $F\cup L$, contains the two corresponding directed arcs.
By Lemma~\ref{lem:matroid_intersection}, we can find a minimum-weight set $D\subseteq A$ of directed arcs that enters every set $\emptyset\neq R\subseteq V\setminus \{r\}$ at least twice. 
W.l.o.g. we can assume that $D$ contains all the arcs corresponding to $F$ (since they have weight zero).

Next, we construct a WTAP instance together with an up-link solution, to which we will later apply Lemma~\ref{lem:wtap}.
Let $S$ be the set of links corresponding to directed edges in $D$.
Here, we define $S$ to contain only a single copy of a link $\{u,v\}$ even if $D$ contains both $(u,v)$ and $(v,u)$.
The lemma below shows that $F\cup S$ contains a cheap spanning tree, which will be the input tree of the WTAP instance we construct.

\begin{lemma}\label{lem:cheap_arb}
The graph $(V,F\cup S)$ contains a spanning tree $(V,F\cup S_{tree})$ with $|S_{tree}|= \nc - 1$.
\end{lemma}
\begin{proof}
First, we observe that $(V,D)$ contains a spanning arborescence.
This follows from the fact that every nonempty set of vertices that does not contain $r$ has an entering arc in $D$ and thus every vertex is reachable from $r$ in $(V,D)$.
Thus, $(V,F\cup S)$ is connected.
Because $(V,F)$ is a forest, there exists a spanning tree containing all edges from $F$.
This spanning tree has weight (i.e. number of links) $|V|- 1 - |F| = \nc - 1$.
\end{proof}

Let $(V,F\cup S_{\mathrm{tree}})$ be a spanning tree as in Lemma~\ref{lem:cheap_arb} and
let $S_{\mathrm{tap}}\coloneqq S \setminus S_{\mathrm{tree}}$.
Moreover, let $D_{\mathrm{tree}} \subseteq D$ and $D_{\mathrm{tap}}\subseteq D$ be the arcs in $D$ whose underlying undirected edges are contained in $F\cup S_{\mathrm{tree}}$ and $S_{\mathrm{tap}}$, respectively. 

Because the choice of $D$ implies that $F\cup S_{\mathrm{tree}} \cup S_{\mathrm{tap}}$ is 2-edge-connected, $S_{\mathrm{tap}}$ is a feasible solution for the WTAP instance with tree $(V,F\cup S_{\mathrm{tree}})$.
Moreover, we obtain the following even stronger property, the proof of which is inspired by an observation from \cite{CTZ21}.

\begin{lemma}\label{lem:cheap_up_link_solution}
Consider the WTAP instance with tree $G=(V,F\cup S_{\mathrm{tree}})$ and link set consisting of all shadows of links in $S_{\mathrm{tap}}$.
Then one can find in polynomial time, a feasible solution for this instance that consists only of up-links and has weight (i.e., number of links) at most $|S_{\mathrm{tap}}|$.
\end{lemma}
\begin{proof}
We call a cut, a 1-cut of the spanning tree $G$, if it contains only a single edge of $G$.
Recall that we can view WTAP as the problem of covering the 1-cuts of $G$ by links.
First, we show that for every 1-cut $\delta(R)$ with $\emptyset\neq R\subseteq V\setminus \{r\}$, there is a directed arc in $D_{\mathrm{tap}}$ that enters $R$.
To prove this, we first observe that $D$ contains at least two arcs entering $R$.
One of these arcs might be an arc $(a,b)\in D_{\mathrm{tree}}$ corresponding to the unique edge $\{a,b\} \in (F\cup S_{\mathrm{tree}}) \cap \delta(R)$.
However, because $\delta(R)$ is a 1-cut of the tree $G$, the other arc in $D$ that enters $R$ must be an element of $D_{\mathrm{tap}}=D\setminus D_{\mathrm{tree}}$.

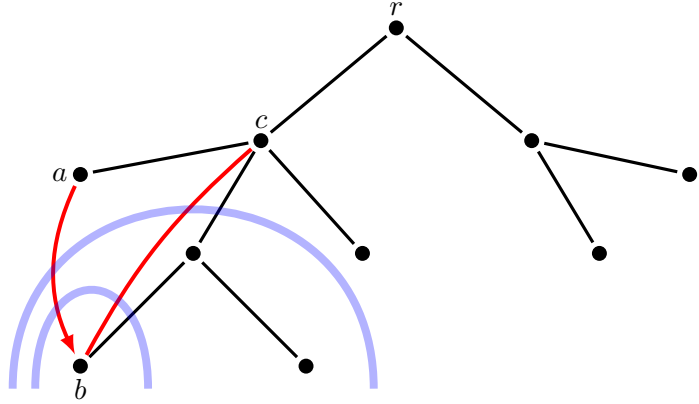
\begin{figure}
\begin{center}
\begin{tikzpicture}[scale=1.5]

\tikzset{
vertex/.style={fill=black,circle,minimum size=0pt, inner sep=2pt, outer sep=2pt}
}

\tikzset{
lks/.style={line width=1.5pt}
}

\begin{scope}[every node/.style=vertex]

\node (1) at (3,3) {};
\node (1a) at (4.2,2) {};
\node (1b) at (5.6,1.7) {};
\node (1c) at (4.8,1) {};

\node (2) at (1.8,2) {};
\node (2a) at (0.2,1.7) {};
\node (2b) at (2.7,1) {};
\node (3) at (1.2,1) {};
\node (3a) at (2.2,0) {};
\node (4) at (0.2,0) {};

\begin{scope}[very thick]

\draw (1) -- (2);
\draw (2) -- (3);
\draw (3) -- (4);

\draw (1) -- (1a) -- (1b);
\draw (1a) -- (1c);

\draw (2a) -- (2) -- (2b);

\draw (3) --(3a);

\end{scope}
\end{scope}

\node[above=1pt] () at (1) {$r$};
\node[left=1pt] () at (2a) {$a$};
\node[above=1pt] () at (2) {$c$};
\node[below=1pt] () at (4) {$b$};

\begin{scope}[lks]
\draw[bend right=25, red, ->, >=latex] (2a) to (4);
\draw[bend right=10, red] (2) to (4);
\end{scope}

\begin{scope}[blue, line width=3pt, opacity=0.3]
\draw[bend left=90, looseness=1.7] (-0.4,-0.2) to (2.8,-0.2);
\draw[bend left=90, looseness=3] (-0.2,-0.2) to (0.8,-0.2);
\end{scope}
\end{tikzpicture}
\end{center}
\caption{
Illustration of the proof of Lemma~\ref{lem:cheap_up_link_solution}.
The spanning tree $S_{\mathrm{tree}}$ is shown in black.
In red, there is a directed link $(a,b)$ and the corresponding up-link $\{c,b\}$.
The cuts $\delta(R)$ with $\emptyset \ne R\subseteq V\setminus \{r\}$ that contain precisely one edge of the tree $S_{\mathrm{tree}}$ and fulfill $(a,b)\in \delta^-(R)$, are shown in light blue.
For each such cut, we have $\{b,c\} \in \delta(R)$.
\label{fig:up_links_correspond_to_directed_links}
}
\end{figure}

To construct the desired up-link solution, we replace every directed link $(a,b) \in D_{\mathrm{tap}}$ by the up-link $\{c,b\}$, where $c$ is the lowest common ancestor of $a$ and $b$ in the tree $G$ (with respect to the root $r$).
See Figure~\ref{fig:up_links_correspond_to_directed_links}.
Note that $\{c,b\}$ is a shadow of the link $\{a,b\}$.

Now consider any 1-cut $\delta(R)$ with $R\subseteq V\setminus \{r\}$.
We observe that if the directed link $(a,b)$ enters $R$, then also the up-link $\{c,b\}$ covers this 1-cut, i.e., $\{c,a\}\in \delta(R)$.
Indeed, because $\delta(R)$ is a 1-cut of $G$ and $R$ does not contain the root, the fact that $b\in R$ and $a\notin R$ implies that the lowest common ancestor $c$ of $a$ and $b$ is not contained in $R$.
\end{proof}

We are now ready to prove Lemma \ref{lem:result_tap_algo}.
\begin{proof}[Proof of Lemma \ref{lem:result_tap_algo}]
By Lemma~\ref{lem:matroid_intersection}, we can compute the sets $D$ and $S$ in polynomial time.
Then by Lemma~\ref{lem:cheap_arb}, we can compute a spanning tree $(V, F\cup S_{\mathrm{tree}})$ in $(V,F \cup S)$ of weight $\nc -1$.

Clearly, the set $\OPT \setminus S_{\mathrm{tree}}$ is a feasible TAP solution to this instance with at most $\opt$ many links.
In particular, the optimal solution value of the TAP instance $((V,F\cup S_{\mathrm{tree}}), L \setminus  S_{\mathrm{tree}})$ is at most the optimal solution value $\opt$ for the original FAP instance.
Moreover, by Lemma~\ref{lem:cheap_up_link_solution}, we can construct an up-link solution for this instance with at most $|S_{\mathrm{tap}}|$ many links.
Using Lemma~\ref{lem:2apx}, we get
\[
|S_{\mathrm{tap}}|\ =\ |S| - |S_{\mathrm{tree}}| = |S| - \nc + 1\ \le\ |D| - \nc + 1 \ \le\ 2 \cdot \opt - \nc + 1,
\]
and hence applying the algorithm from Lemma~\ref{lem:wtap} yields a TAP solution with at most 
\[
 \left( 1 + \ln\Big(2-\frac{\nc - 1}{\opt}\Big) + \epsilon\right) \cdot \opt
 \]
  many links.
The union of these links and the $\nc -1$ links in $S_{\mathrm{tree}}$ is the desired FAP solution. The claim follows by observing that replacing $\nc -1$ by $\nc$ only weakens the overall bound.
\end{proof}

\section{Algorithm for Forests with Many Connected Components}\label{sec:many_components}

In this section we prove Lemma~\ref{lem:result_credit_algo}. From the discussion in Section \ref{sec:overview}, this reduces to proving Lemma \ref{lem:result_pap_algo}, i.e., to providing a $(\tfrac{7}{4}, \tfrac{7}{4})$-approximation algorithm for PAP without isolated nodes.
Therefore, in what follows we will assume that every connected component of $(V,F)$ is a path of length at least $1$.

\subsection{Overview of Our Algorithm}

Let us start with a brief recap of how our algorithm and its analysis work on a high level.
In our algorithm we maintain a partial solution $S\subseteq L$, which is not necessarily a feasible solution.
In the analysis, we maintain in addition a fractional set of (nonnegative) credits, which are distributed over the current graph $H=(V, F\cup S)$.
Credits will be assigned to some of the vertices, edges, connected components, and 2-edge-connected components.
Throughout the algorithm we maintain the invariant that the number $|S|$ of links in our partial solution plus the total number of credits $\cred(H)$ of $H$ is at most $\tfrac{7}{4} \opt + \tfrac{7}{4}  \left(\opt-\nc\right)$.
At the end of the algorithm, $S$ will be a feasible solution, i.e., $H=(V,F\cup S)$ will be 2-edge-connected.
Then our invariant (and the fact that credits are nonnegative) implies $|S|\le \tfrac{7}{4} \opt + \tfrac{7}{4} \cdot \left(\opt -\nc\right)$.

Our algorithm consists of three steps.
First, we use a matching algorithm to compute an initial partial solution $S$.
The goal of the second step, called bridge covering, is to ensure that every connected component of $H\coloneqq (V,F\cup S)$ is 2-edge-connected.
We describe this step in Section~\ref{sec:bridge_covering}.
Finally, in the third step, called gluing, we ensure that $H$ becomes 2-edge-connected (see Section~\ref{sec:gluing}).
Algorithm~\ref{algo:pap_overview} describes our overall algorithm for PAP without isolated nodes.

\begin{algorithm}[H]
\caption{Algorithm for PAP without isolated nodes \label{algo:pap_overview}}

\begin{enumerate}
\item \textbf{Initialization:}\\
Let $M \subseteq L$  be a maximum cardinality matching on the leaves of $(V,F)$ that contains no bad link, i.e., no link with both endpoints in the same path in $(V,F)$. \\
 Initialize $S\coloneqq M$.
\item \textbf{Bridge covering:}\\
As long as $H=(V, F\cup S)$ has a connected component that is not 2-edge-connected, iterate the following:
\begin{itemize}\itemsep0pt
\item Let $C$ be a connected component of $H$ that is not 2-edge-connected and let $x$ be a leaf of the tree $T^C$  (arising from the contraction of each 2-edge-connected component of $C$).
\item Let $z$ be the vertex of the tree $T^C$ that has maximum distance from $x$ (in $T^C$) among all vertices that are reachable from $x$ by an alternating trail (as defined in Section~\ref{sec:bridge_covering}).
\item Augment $S$ along an alternating $x$-$z$ trail (as defined in Section~\ref{sec:bridge_covering}).
\end{itemize}
\item \textbf{Gluing:}\\
As long as $H = (V, F\cup S)$ is not 2-edge-connected, iterate the following:
\begin{itemize}\itemsep0pt
\item If there is a good cycle $Q$, then glue $H$ along $Q$ (definitions in Section~\ref{sec:gluing}).
\item Otherwise take an arbitrary cycle $Q$ in the graph $G_H$ (arising from the contraction of the $2$-edge-connected components of $H$) and add the links of $Q$ to $S$.
\end{itemize}
\item Return $S$.
\end{enumerate}
\end{algorithm}

\subsection{Credit Scheme and Invariants}

In this section we show how we distribute credits and describe the invariants that we maintain during our algorithm.
We emphasize that we will use credits only in the analysis of our algorithm.
Therefore, the distribution of credits can depend on a fixed (but unknown) optimal solution $\OPT$. 

To define our credit scheme, we need the notion of \emph{simple components}.

\begin{definition}[simple components]\label{def:simple}
A connected component of $H=(V,F\cup S)$ is called \emph{simple} if it is a cycle that contains exactly 2 links.
\end{definition}

Simple components will play a special role in the gluing step because they are the components that prevent us from only adding links (and never removing links) in the gluing step.
We remark that in particular in the example in Figure~\ref{fig:removing_needed_gluing} all connected components of $H=(V,F\cup M)$ are simple.

Let us now define the credit invariants.
Consider any point of the algorithm, where we have a current partial solution $S\subseteq L$ and define the current subgraph to be $H\coloneqq(V,F\cup S)$.
Recall that by assumption, the graph $H$ has no isolated vertices.
We call a 2-edge-connected component of $H$ \emph{nontrivial} if it contains at least two vertices.
The \emph{2EC-blocks} of a connected component of $H$ are its inclusionwise maximal 2-edge-connected subgraphs that are nontrivial.

\begin{definition}[lonely vertex]
If a vertex $v\in V$ does not belong to any nontrivial 2-edge-connected component of $H$, we call $v$ a \emph{lonely} vertex (of $H$). 
\end{definition}

For $H=(V,F\cup S)$, we assign credits according to the following rules, where we set $\epsilon \coloneqq \tfrac{1}{4}$.
\begin{enumerate}[label=(\Alph*)]\itemsep0pt
\item\label{item:vertex_credits} Vertices receive the following credits.
\begin{enumerate}[label=(A\arabic*)]\itemsep0pt
\item \label{item:leaf_credits} Every leaf of $H$, i.e., every vertex $v$ with $|\delta_{F\cup S}(v)|=1$, receives $1$ credit. 
\item\label{item:lonely_credits} A vertex $v$ of $H$, receives $\tfrac{1}{2}(|\delta_{\OPT \cup F}(v)|-2)$ additional credits if it is a lonely vertex or it belongs to a simple component.
\end{enumerate}
\item\label{item:bridge_credits} Every bridge $\ell \in S$, i.e., every link $\ell\in S$ that is not part of a 2-edge-connected component of $H$, receives $1-\eps$ credits.
\item\label{item:component_credits} Every connected component of $H$ that contains at least one bridge receives $1$ credit. Every connected component of $H$ that is $2$-edge-connected receives $2-2\eps$ credit if it simple and $2$ credits otherwise.
\item\label{item:2ec_component_credits} Every 2EC-block of a connected component that contains bridges receives $1$ credit.
\end{enumerate}
We remark that the credits assigned according to \ref{item:leaf_credits} and \ref{item:lonely_credits} add up. 
We also observe that the credits assigned according to \ref{item:lonely_credits} depend on the (unknown) optimal solution $\OPT$.
These \emph{implicit credits} are needed in order to be able to compare to $\opt$ because using only a lower bound on $\opt$ that is based on the number of unmatched leaves in $M$\footnote{We will give such a bound in Lemma~\ref{lem:matching_bound}.} is not sufficient to achieve any finite approximation guarantee, as one can see from the example in Figure~\ref{fig:implicit_credits_needed}.

\begin{figure}[H]
\begin{center}
\begin{tikzpicture}[scale=1.5]

\tikzset{
vertex/.style={fill=black,circle,minimum size=0pt, inner sep=2pt, outer sep=2pt}
}
\tikzset{
block/.style={draw=black, very thick, circle,minimum size=0pt, inner sep=2pt, outer sep=2pt}
}

\begin{scope}[every node/.style={vertex}]
\node (1) at (1,1) {};
\node (2) at (2,1) {};
\node (3) at (3,1) {};
\node (4) at (4,1) {};
\node (5) at (5,1) {};
\node (6) at (6,1) {};
\node (7) at (7,1) {};
\node (8) at (8,1) {};
\node (9) at (9,1) {};
\node (10) at (10,1) {};

\end{scope}

\begin{scope}[blue, very thick]
\draw[bend left] (1) to (3);
\draw[bend right] (2) to (5);
\draw[bend right] (6) to (9);
\draw [bend left] (4) to (7);
\draw [bend left] (8) to (10);
\end{scope}

\begin{scope}[very thick, densely dashed]
\draw (1) -- (2);
\draw (3) -- (4);
\draw (5) -- (6);
\draw (7) -- (8);
\draw (9) -- (10);
\draw (2) -- (3);
\draw (4) -- (5);
\draw (6) -- (7);
\draw (8) -- (9);
\end{scope}

\end{tikzpicture}
\end{center}
\caption{
Example showing that a lower bound based on the number of unmatched leaves in the initial matching $M$ (which would be empty in this example) does not lead to a strong enough lower bound, which is why we introduce implicit credits.
The dashed edges are edges of the forest $F$, while blue edges are links in $\OPT$.
Then, indeed, the number of unmatched leaves of the forest is $2$ (and thus constant), but $\opt$ can be unbounded (by increasing the size of the above example in the obvious way).
\label{fig:implicit_credits_needed}
}
\end{figure}
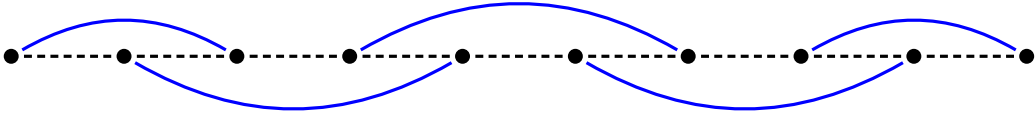

 We use $\cred(H)$ to denote the total number of credits of $H\coloneqq(V,F\cup S)$ according to the above rules \ref{item:vertex_credits} -- \ref{item:2ec_component_credits}.
To prove that our algorithm fulfills the desired approximation guarantee, we prove that we maintain Invariant \ref{invariant:credits}, which in terms of $\eps$ can be restated as
$$
\cred(H) + |S| \ \le\ \left(2-\epsilon \right) \cdot \opt +  (1+ 3 \epsilon)\cdot \left(\opt -\nc\right).
$$ 

To this end, we will first show that Invariant~\ref{invariant:credits} is fulfilled for the initial choice of $S=M$.
Then we provide a procedure that maintains this invariant and at the end of which $H$ is 2-edge-connected.
In order to show that the invariant is maintained, we will show that the total $\cred(H)+|S|$ never increases.

Besides Invariant~\ref{invariant:credits}, we will also maintain the following invariant.

\begin{invariant}\label{invariant:lonely_vertices}
Every lonely vertex of $H$ has degree at most $2$ in $H$.
\end{invariant}

\subsection{Invariants after the Initialization}

Let $M \subseteq L$  be a maximum cardinality matching on the leaves of $(V,F)$ that contains no link with both endpoints in the same path (we call such links \emph{bad}). Our initial solution is $S=M$ and we let $H=(V,F\cup M))$. We next show that $H$ fulfills Invariants~\ref{invariant:credits} and~\ref{invariant:lonely_vertices}. 
To this aim, we first relate the number of unmatched leaves of $(V,F)$, which is the number of leaves of the initial graph $H=(V,F\cup M)$, to the size $\opt$ of an optimal solution.

\begin{lemma}\label{lem:matching_bound}
The number $2\nc-2|M|$ of \emph{unmatched} leaves, i.e., the number of leaves of $(V,F)$ that do not have an incident edge in $M$, is at most $4\cdot \left(\opt-\nc\right)$.
\end{lemma}
\begin{proof}
Let $\OPT \subseteq L$ be an optimal solution and let 
$M_{\OPT}\subseteq \OPT$ be a maximal matching in $\OPT$ that contains no bad links.
Then $|M|\ge |M_{\OPT}|$ and thus it suffices to show that  at most $4\cdot (\opt-\nc)$ leaves of the forest $(V,F)$ have no incident edge in the matching $M_{\OPT}$.

First, we observe that every connected component of the forest $(V,F)$ has at least two edges incident to it that have their other endpoint in a different connected component.
Because bad links have both endpoints in the same connected component of the forest, this implies 
\begin{equation}\label{eq:first_lower_bound_opt}
\opt \ge \nc + \opt_{\mathrm{bad}},
\end{equation}  
where $\opt_{\mathrm{bad}}$ denotes the number of bad links in $\OPT$.
Let $V_{\mathrm{leaf}}\subseteq V$ denote the set of leaves of the forest $(V,F)$.
Because every connected component of this forest is a path, we have
\[
 2\cdot \nc =\ \big|V_{\mathrm{leaf}} \big|\ =\ 2\cdot |M_{\OPT}| +  \big|\big\{ v\in V_{\mathrm{leaf}} :  M_{\OPT} \cap \delta(v)= \emptyset \big\}\big|.
\]
Using the maximality of $M_{\OPT}$ and the fact that every leaf of $(V,F)$ is the endpoint of a link in $\OPT$, this implies
\begin{equation}\label{eq:second_lower_bound_opt}
\begin{aligned}
\opt \ge&\  |M_{\OPT}| + \big|\big\{ v\in V_{\mathrm{leaf}} : M_{\OPT} \cap \delta(v)= \emptyset \big\}\big| - \opt_{\mathrm{bad}} \\
=&\ \nc  + \tfrac{1}{2} \cdot \big|\big\{ v\in V_{\mathrm{leaf}} : M_{\OPT} \cap \delta(v)= \emptyset \big\}\big| - \opt_{\mathrm{bad}}.
\end{aligned}
\end{equation}
Adding up \eqref{eq:first_lower_bound_opt} and \eqref{eq:second_lower_bound_opt}, we obtain
\[
2\cdot (\opt-\nc) \ge \tfrac{1}{2} \cdot  \big|\big\{ v\in V_{\mathrm{leaf}} : M_{\OPT} \cap \delta(v)= \emptyset \big\}\big|. \qedhere
\]
\end{proof}

We now show that $H=(V,F\cup M))$ fulfills Invariants~\ref{invariant:credits} and~\ref{invariant:lonely_vertices}. 
 
\begin{lemma}\label{lem:initialization_invariants}
$H=(V,F\cup M)$ satisfies Invariants \ref{invariant:credits} and \ref{invariant:lonely_vertices}.
\end{lemma}
\begin{proof}
Invariant \ref{invariant:lonely_vertices} follows from the fact that all vertices have degree at most $2$ in $F$ and this property is maintained by adding a matching on the leaves of $(V,F)$. 

Consider next Invariant \ref{invariant:credits}. First, we observe that every connected component of $H$ is either a path or a cycle.
Moreover, because $M$ contains no bad links, every cycle in $H$ contains at least two links from $M$.
Let us now give upper bounds on $\cred(H)$.
\begin{itemize}
\item The number of credits due to \ref{item:leaf_credits} is the number of leaves of $H$, i.e.,
$2 \nc -2 |M|$. 
Therefore, we can upper bound this number of credits by
\begin{align*}
2\nc-2|M| \overset{Lem. \ref{lem:matching_bound}}{\leq} &\ \big(1-\tfrac{\epsilon}{2}\big) \cdot \big(2 \cdot \nc - 2|M|\big) + \tfrac{\epsilon}{2} \cdot \big( 4\cdot (\opt -\nc)\big) \\
=&\ (2-\epsilon) \cdot (\nc-|M|) + 2\epsilon \cdot (\opt-\nc).
\end{align*}

\item The number of credits according to \ref{item:lonely_credits} is at most
\[
\tfrac{1}{2} \sum_{v\in V} \big(|\delta_{\OPT \cup F}(v)|-2) \ =\ \opt + |F| - |V| \ =\ \opt -\nc. 
\]
\item The number of credits according to \ref{item:bridge_credits} is 
\[
(1-\epsilon)\cdot |\{\ell\in M : \ell\text{ is a bridge of }H\}|.
\]
\item Let us now consider credit rule \ref{item:component_credits}.
The number of connected components of $H$ that contain bridges is  $\tfrac{1}{2}$ times the number of leaves of $H$, because each such connected component is a path. 
Thus, by Lemma~\ref{lem:matching_bound}, the connected components of $H$ that contain bridges have at most
\[
2\cdot (\opt-\nc)
\]
credits in total. Recall that every connected component that is 2-edge-connected is a cycle and contains at least two links from $M$.
If it contains exactly two links from $M$ it is simple.
Thus, because $3\cdot (1-\epsilon) \ge 2$, the number of credits of a connected component that is 2-edge-connected can be upper bounded by $(1-\epsilon)$ times the number of links it contains (from $M$).
Because these links are no bridges we can upper bound the total number of credits of connected components that are 2-edge-connected by
\[
 (1-\epsilon) \cdot \big(|M| - |\{\ell\in M : \ell\text{ is a bridge of }H\}|\big)
\]
\item Finally, because every connected component of $H$ is either a cycle or a path, there are no 2EC-blocks in connected components containing bridges.
Thus, there are no credits due to \ref{item:2ec_component_credits}.
\end{itemize}
Summing up, we get
\begin{align*}
\cred(H) \ \le&\ (2-\epsilon) \cdot (\nc-|M|) + (3+ 2\epsilon) \cdot (\opt-\nc) + (1-\epsilon)\cdot |M|.
\end{align*}
Thus, $S=M$ implies
\begin{align*}
\cred(H) + |S| \ \le&\ (2-\epsilon) \cdot (\nc-|M|) + (3+ 2\epsilon) \cdot (\opt-\nc) + (2-\epsilon)\cdot |M| \\
=&\ (2-\epsilon) \cdot \opt + (1+ 3\epsilon) \cdot (\opt-\nc).
\end{align*}

\end{proof}

\subsection{Bridge Covering using Alternating Trails}\label{sec:bridge_covering}

In this section we describe and analyze the bridge covering procedure.
We use it to augment $S$ in such a way that after the bridge covering step every connected component of $H=(V,F\cup S)$ is 2-edge-connected.
In the following we assume that $H$ has at least one bridge and describe an augmentation procedure that decreases the number of bridges of $H$ by at least $1$.
We can then iteratively apply this procedure until every connected component of $H=(V,F\cup S)$ is 2-edge-connected.
\bigskip

Let $C$ be a connected component of $H$ that is not 2-edge-connected.
We aim at ``covering'' at least one bridge of the component $C$.
To this end we might use paths ``going through'' different connected components, but it is irrelevant for us what ``happens inside these components''.
Thus, it is useful to consider the graph where these components are contracted.
Moreover, we will not add or remove any links inside 2EC-blocks of $C$ and thus we will also contract them.

Let $G^C$ and $H^C$ be the (multi-)graphs obtained from $(V,F\cup L)$ and $H$, respectively, by contracting every connected component except for $C$ and contracting every 2EC-block of $C$.
Then $H^C$ is the union of a tree $T^C$ with at least one edge (resulting from $C$ by contracting its 2EC-blocks) and singletons (resulting from the contraction of connected components distinct from $C$).
In the following we identify the edges/links in $(V,F\cup L)$ and the corresponding edges/links in $G^C$.

We will aim at ``covering as many edges/bridges of the tree $T^C$ as possible'' by an \emph{alternating trail} starting in $x$. 
See Figure~\ref{fig:alternating_trail} for an example.

We use the following definitions.
A \emph{trail} is a walk that uses every edge at most once.
A \emph{path} is a walk that visits every vertex at most once.

\begin{definition}(alternating trail)\label{def:alternating_trail}
Let $z$ be a vertex of $T^C$ that is distinct from $x$.
We denote the $x$-$z$ path in $T^C$ by $P^{xz}$ and the links in $S\cap P^{xz}$ by $\ell_1, \dots, \ell_l$ in the order they appear on the path $P^{xz}$ (starting from the link closest to $x$).
We denote the endpoints of the link $\ell_i$ by $u_i$ and $v_i$, where $u_i$ is the endpoint closer to $x$ in the path $P^{xz}$.

An \emph{alternating $x$-$z$-trail} is a trail $P\subseteq L$ in $G^C$ that has the following properties:

\begin{itemize}\itemsep0pt
\item $P$ starts in $x$.
\item $P$ ends in the vertex $z\ne x$ in $T^C$.
\item $P$ consists alternatingly of 
\begin{itemize}
\item a path for which all interior vertices are not contained in the tree $T^C$, and 
\item a link $\ell_i =\{u_i,v_i\}$ for some $i\in\{1,\dots,l\}$, where $P$ visits $v_i$ before $u_i$.
\end{itemize}

\item $P$ visits every vertex outside of $T^C$ at most once.
\item The links in $P\cap P^x$ are visited in an order of increasing distance from $x$, i.e., for $i<j$ with $\ell_i,\ell_j \in P$, the link $\ell_i$ appears before $\ell_j$ on $P$.
\end{itemize}
\end{definition}
Note that the last link of an alternating $x$-$z$ trail $P$ is always a link in $L\setminus P^{xz}$.
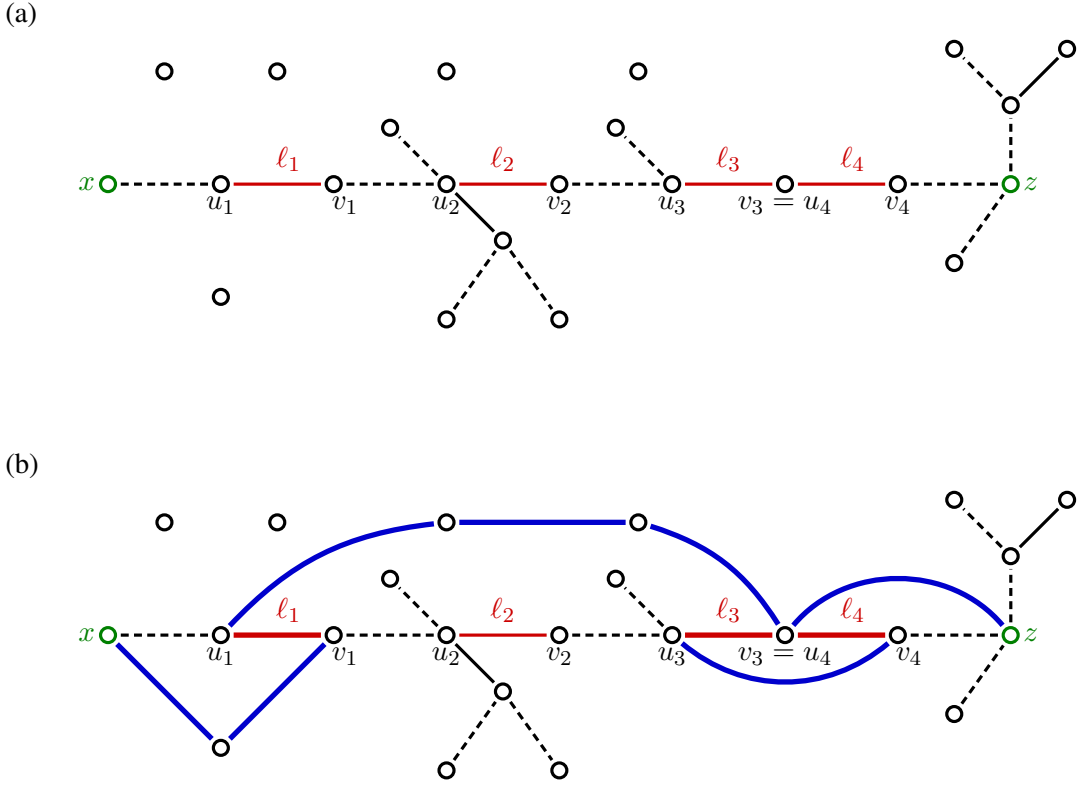
\begin{figure}[h]
\begin{center}
\begin{tikzpicture}[scale=1.5]

\tikzset{vertex/.style={draw=black, very thick, circle,minimum size=0pt, inner sep=2pt, outer sep=2pt}
}

\node[right] () at (-1,1.5) {(a)};

\begin{scope}[every node/.style={vertex}]

\node[green!50!black] (0) at (0,0) {};
\node (1) at (1,0) {};
\node (2) at (2,0) {};
\node (3) at (3,0) {};
\node (4) at (4,0) {};
\node (5) at (5,0) {};
\node (6) at (6,0) {};
\node (7) at (7,0) {};
\node[green!50!black] (8) at (8,0) {};

\node (3a) at (2.5,0.5) {};
\node (3b) at (3.5,-0.5) {};
\node (3c) at (3,-1.2) {};
\node (3d) at (4,-1.2) {};
\node (5a) at (4.5,0.5) {};
\node (8a) at (7.5, -0.7) {};
\node (8b) at (8, 0.7) {};
\node (8c) at (7.5, 1.2) {};
\node (8d) at (8.5, 1.2) {};

\node (a) at (1,-1) {};
\node (b) at (3,1) {};
\node (c) at (4.7,1) {};
\node (d) at (0.5,1) {};
\node (e) at (1.5,1) {};

\end{scope}

\node[left=1pt, green!50!black] (x) at (0) {$x$};
\node[right=1pt, green!50!black] (z) at (8) {$z$};

\begin{scope}[red!80!black, very thick]
\draw (2) to (1);
\draw (3) -- (4);
\draw (6) to (5);
\draw (7) to (6);
\end{scope}

\node[above=1pt, red!80!black] () at (1.6,0) {$\ell_1$};
\node[above=1pt, red!80!black] () at (3.5,0) {$\ell_2$};
\node[above=1pt, red!80!black] () at (5.5,0) {$\ell_3$};
\node[above=1pt, red!80!black] () at (6.6,0) {$\ell_4$};

\begin{scope}[very thick]
\draw (8b) -- (8d);
\draw (3) -- (3b);
\end{scope}

\node[below=1pt] () at (1) {$u_1$};
\node[below=1pt] () at (2.1,0) {$v_1$};
\node[below=1pt] () at (3) {$u_2$};
\node[below=1pt] () at (4) {$v_2$};
\node[below=1pt] () at (5) {$u_3$};
\node[below=1pt] () at (6) {$v_3 = u_4$};
\node[below=1pt] () at (7) {$v_4$};

\begin{scope}[very thick, densely dashed]
\draw (0) -- (1);
\draw (2) -- (3);
\draw (3) -- (3a);
\draw (3c) -- (3b) -- (3d);
\draw (5) -- (4);
\draw (5) -- (5a);
\draw (7) -- (8);
\draw (8a) -- (8) -- (8b) -- (8c);
\end{scope}

\begin{scope}[shift={(0, -4)}]

\node[right] () at (-1,1.5) {(b)};

\begin{scope}[every node/.style={vertex}]

\node[green!50!black] (0) at (0,0) {};
\node (1) at (1,0) {};
\node (2) at (2,0) {};
\node (3) at (3,0) {};
\node (4) at (4,0) {};
\node (5) at (5,0) {};
\node (6) at (6,0) {};
\node (7) at (7,0) {};
\node[green!50!black] (8) at (8,0) {};

\node (3a) at (2.5,0.5) {};
\node (3b) at (3.5,-0.5) {};
\node (3c) at (3,-1.2) {};
\node (3d) at (4,-1.2) {};
\node (5a) at (4.5,0.5) {};
\node (8a) at (7.5, -0.7) {};
\node (8b) at (8, 0.7) {};
\node (8c) at (7.5, 1.2) {};
\node (8d) at (8.5, 1.2) {};

\node (a) at (1,-1) {};
\node (b) at (3,1) {};
\node (c) at (4.7,1) {};
\node (d) at (0.5,1) {};
\node (e) at (1.5,1) {};

\end{scope}

\node[left=1pt, green!50!black] (x) at (0) {$x$};
\node[right=1pt, green!50!black] (z) at (8) {$z$};

\begin{scope}[blue!80!black, line width=2pt]
\draw (0) to (a);
\draw (a) to (2);
\draw[bend left=20] (1) to (b);
\draw (b) to (c);
\draw[bend left=20] (c) to (6);
\draw[bend right=40] (5) to (7);
\draw[bend left=50] (6) to (8);
\end{scope}

\begin{scope}[red!80!black, very thick]
\draw[line width=2pt] (2) to (1);
\draw (3) -- (4);
\draw[line width=2pt] (6) to (5);
\draw[line width=2pt] (7) to (6);
\end{scope}

\node[above=1pt, red!80!black] () at (1.6,0) {$\ell_1$};
\node[above=1pt, red!80!black] () at (3.5,0) {$\ell_2$};
\node[above=1pt, red!80!black] () at (5.5,0) {$\ell_3$};
\node[above=1pt, red!80!black] () at (6.6,0) {$\ell_4$};

\begin{scope}[very thick]
\draw (8b) -- (8d);
\draw (3) -- (3b);
\end{scope}

\node[below=1pt] () at (1) {$u_1$};
\node[below=1pt] () at (2.1,0) {$v_1$};
\node[below=1pt] () at (3) {$u_2$};
\node[below=1pt] () at (4) {$v_2$};
\node[below=1pt] () at (5) {$u_3$};
\node[below=1pt] () at (6) {$v_3 = u_4$};
\node[below=1pt] () at (7.1,0) {$v_4$};

\begin{scope}[very thick, densely dashed]
\draw (0) -- (1);
\draw (2) -- (3);
\draw (3) -- (3a);
\draw (3c) -- (3b) -- (3d);
\draw (5) -- (4);
\draw (5) -- (5a);
\draw (7) -- (8);
\draw (8a) -- (8) -- (8b) -- (8c);
\end{scope}

\end{scope}%

\end{tikzpicture}
\end{center}
\caption{
Picture (a) shows the graph $G^C$ with a leaf $x$ of $T^C$ and a vertex $z$ in $T^C$.
Picture (b) shows an alternating $x$-$z$ trail (bold edges).
Solid edges are links from $L$, dashed edges are edges of the forest $F$.
Links in $S\cap P^{xz}$ are red; edges in $P\setminus P^{xz}$ are shown in blue.
\label{fig:alternating_trail}
}
\end{figure}
We will find an alternating trail $P$ that start in $x$ and ends in a vertex $z$ of $P^x$, where $z$ is as far away from $x$ (in $T^C$) as possible.
Then we will \emph{augment} $S$ along the alternating trail $P$.
See Figure~\ref{fig:augmenting_path} for an example. 

\begin{definition}[augmenting along an alternating trail]
If we \emph{augment} the current partial solution $S\subseteq L$ of links by an alternating $x$-$z$ trail $P$, this means that we replace $S$ by $S\triangle P$.
In other words, we remove the links in $P\cap S = P\cap P^{xz}$ from $S$ and add the links in $P\setminus S = P\setminus P^{xz}$ to $S$.
\end{definition}

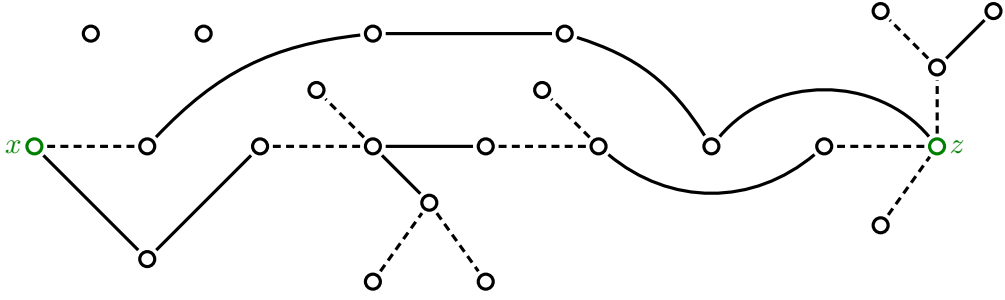
\begin{figure}[h]
\begin{center}
\begin{tikzpicture}[scale=1.5]

\tikzset{vertex/.style={draw=black, very thick, circle,minimum size=0pt, inner sep=2pt, outer sep=2pt}
}

\begin{scope}[every node/.style={vertex}]

\node[green!50!black] (0) at (0,0) {};
\node (1) at (1,0) {};
\node (2) at (2,0) {};
\node (3) at (3,0) {};
\node (4) at (4,0) {};
\node (5) at (5,0) {};
\node (6) at (6,0) {};
\node (7) at (7,0) {};
\node[green!50!black] (8) at (8,0) {};

\node (3a) at (2.5,0.5) {};
\node (3b) at (3.5,-0.5) {};
\node (3c) at (3,-1.2) {};
\node (3d) at (4,-1.2) {};
\node (5a) at (4.5,0.5) {};
\node (8a) at (7.5, -0.7) {};
\node (8b) at (8, 0.7) {};
\node (8c) at (7.5, 1.2) {};
\node (8d) at (8.5, 1.2) {};

\node (a) at (1,-1) {};
\node (b) at (3,1) {};
\node (c) at (4.7,1) {};
\node (d) at (0.5,1) {};
\node (e) at (1.5,1) {};

\end{scope}

\node[left=1pt, green!50!black] (x) at (0) {$x$};
\node[right=1pt, green!50!black] (z) at (8) {$z$};

\begin{scope}[very thick]
\draw (0) to (a);
\draw (a) to (2);
\draw[bend left=20] (1) to (b);
\draw (b) to (c);
\draw[bend left=20] (c) to (6);
\draw[bend right=40] (5) to (7);
\draw[bend left=50] (6) to (8);
\end{scope}

\begin{scope}[ very thick]
\draw (3) -- (4);
\end{scope}

\begin{scope}[very thick]
\draw (8b) -- (8d);
\draw (3) -- (3b);
\end{scope}

\begin{scope}[very thick, densely dashed]
\draw (0) -- (1);
\draw (2) -- (3);
\draw (3) -- (3a);
\draw (3c) -- (3b) -- (3d);
\draw (5) -- (4);
\draw (5) -- (5a);
\draw (7) -- (8);
\draw (8a) -- (8) -- (8b) -- (8c);
\end{scope}

\end{tikzpicture}
\end{center}
\caption{
The result of the augmentation along the alternating trail from Figure~\ref{fig:alternating_trail} (in the graph $G^C$, i.e., before undoing contractions of 2EC-blocks and connected components of $H$).
\label{fig:augmenting_path}
}
\end{figure}

\begin{lemma}
Given a leaf $x$ and an arbitrary vertex $z \ne x$ in the tree $T^C$, we can in polynomial time either find an alternating $x$-$z$ trail or decide that no such trail exists.
\end{lemma}
\begin{proof}
As in Definition~\ref{def:alternating_trail}, we denote the $x$-$z$ path in $T^C$ by $P^{xz}$ and the links in $S\cap P^{xz}$ by $\ell_1, \dots, \ell_l$ in the order they appear on the path $P^{xz}$.
Also, we again denote the endpoints of the link $\ell_i$ by $u_i$ and $v_i$, where $u_i$ is the endpoint closer to $x$ in the path $P^{xz}$.
We construct a directed auxiliary graph with vertex set $\{x, \ell_1,\dots,\ell_l,z\}$ and arcs
\begin{itemize}\itemsep0pt
\item $(\ell_i,\ell_j)$ if $i < j$ and there exists an $u_i$-$v_j$ path in $G^C$ for which all interior vertices are not contained in $T^C$;
\item $(x, \ell_i)$ if there exists an $x$-$v_i$ path in $G^C$ for which all interior vertices are not contained in $T^C$;
\item $(\ell_i,z)$ if  there exists an $u_i$-$z$ path in $G^C$ for which all interior vertices are not contained in $T^C$.
\end{itemize}
Then, by the definition of alternating trails, every alternating $x$-$z$ trail corresponds to a directed $x$-$z$ path in this auxiliary graph.
In particular, if an alternating $x$-$z$ trail exist, $z$ is reachable from $x$ in the directed auxiliary graph.
To find such an alternating trail, we compute a shortest $x$-$z$ path $\overline{P}$ in the auxiliary graph (where shortest path means one with the minimum number of arcs).

We construct $P$ from $\overline{P}$ by replacing every arc $(\ell_i, \ell_j)$ of $\overline{P}$ by a $u_i$-$v_j$ path whose internal vertices are not contained in $T^C$, and similarly replacing arcs $(x,\ell_i)$ and $(\ell_i,z)$ by $x$-$v_i$ and $u_i$-$z$ paths, respectiviely.
We claim that $P$ is an alternating $x$-$z$ trail.
To this end, we need to show that paths corresponding to different arcs of $\overline{P}$ are vertex disjoint, i.e., they have no common interior vertex.
This follows from the fact that $\overline{P}$ is a shortest $x$-$z$ path as one can see as follows.

Suppose there are arcs $(\ell_i,\ell_j)$ and $(\ell_p,\ell_q)$ in $\overline{P}$, appearing without loss of generality in this order on $\overline{P}$, where 
the $u_i$-$v_j$ path and the $u_p$-$v_q$ path corresponding to these arcs are not vertex disjoint.
This would imply that $v_q$ is reachable from $u_i$ by a path whose internal vertices are not part of $T^C$.
Hence, the auxiliary graph contains an arc $(\ell_i,\ell_q)$.
This is a contradiction to $\overline{P}$ being a shortest $x$-$z$ path because we could use the arc $(\ell_i,\ell_q)$ to shortcut it.
A similar argument leads to a contradiction also if one (or both) of the arcs are of the form $(x,\ell_i)$ or $(\ell_i,z)$.
\end{proof}

The next lemma shows that when we augment $S$ long an alternating $x$-$z$ trail $P$, we merge all vertices and 2EC-blocks visited by $P^{xz}$ into a single 2-edge-connected component, which also contains at least one vertex from each of the connected components of $H$ that correspond to the vertices of $P$ outside $T^C$.

\begin{lemma}\label{lem:new_component_bridge_covering}
$P \triangle P^{xz}$ is the edge set of a cycle in $G^C$ containing all vertices visited by $P$ or $P^{xz}$.
\end{lemma}
\begin{proof}
First, we show that all vertices visited by $P$ or $P^{xz}$ have even degree in $P\triangle P^{xz}$.
Indeed, because all such vertices have even degree in the disjoint union of the two $x$-$z$ trails $P$ and $P^{xz}$, they also have even degree in the symmetric difference $P\triangle P^{xz}$. 
(Here we use that $P\triangle P^{xz}$ arises from the disjoint union by removing pairs of parallel edges.)
Moreover, no vertex has degree larger than $2$ in $P\triangle P^{xz}$ because whenever an interior vertex of $P^{xz}$ is visited by $P$, the trail $P$ uses a link $\ell_i \in P\cap P^{xz}$ that is incident to this vertex.
We conclude that all vertices visited by $P$ or $P^{xz}$ have degree $0$ or $2$ in $P\cap P^{xz}$.
Hence, it remains to show that all these vertices are connected in $P\cap P^{xz}$.

To see this, we first observe that every vertex visited by $P$ but not by $P^{xz}$ is connected to some vertex visited by $P^{xz}$ in $P\setminus P^{xz} \subseteq P\triangle P^{xz}$ (by the definition of an alternating trail).
Hence it suffices to shows that every vertex $v$ on $P^{xz}$, except for $x$, is connected to another vertex $v'$ on $P^{xz}$ that is closer to $x$ (on $P^{xz}$). 
This is indeed the case because either $v=v_j$ for some $j\in\{1,\dots,l\}$ with $\ell_i\in P$, in which case a path in $P\setminus P^{xz}$ that precedes $v$ in $P$ connects $v$ to a vertex $v' = u_i$ with $i<j$, or the edge $\{v',v\}$ preceding $v$ on the path $P^{xz}$ (viewed as a path starting at $x$) is contained in $P \triangle P^{xz}$.
\end{proof}

The next lemma will be crucial to prove that Invariant~\ref{invariant:credits} is maintained when augmenting along $P$.
Informally speaking, we use it to show that $\cred(H)$ decreases sufficiently (to pay for the increase of $|S|$) for one of the following reasons:
\begin{itemize}\itemsep0pt
\item our alternating trail reaches a leaf of $T^C$, in which case the required number of credits in $H$ decreases due to either \ref{item:leaf_credits} or \ref{item:2ec_component_credits};
\item a vertex of degree at least $3$ in $T^C$ gets merged into another 2-edge-connected component; this vertex must correspond to a 2EC-block by Invariant~\ref{invariant:lonely_vertices}; thus it had $1$ credit (by \ref{item:2ec_component_credits}) before it was merged;
\item we merge bridges in $D\cap P^{xz}$ or lonely vertices on $P^{xz}$ into a nontrivial 2-edge-connected component, leading to a decrease of the credits in $H$ required by \ref{item:bridge_credits} and \ref{item:lonely_credits}.
\end{itemize}
We remark that the statement of the below lemma is nontrivial only for $q\in\{0,1\}$.
Moreover, we highlight that the vertex set $W$ of the path $P^{xz}$ can contain both original vertices from $V$ and vertices arising from the contraction of a 2EC-block of $C$.

\begin{lemma}\label{lem:opt_credits_lower_bound}
Let $z$ be a vertex that is furthest away from $x$ in $T^C$ among all vertices reachable from $x$ by an alternating trail. Let $P^{xz}$ be the $x$-$z$ path in $T^C$ and let $W$ be the vertex set of $P^{xz}$.
Then at least one of the following holds:
\begin{itemize}\itemsep0pt
\item $z$ is a leaf of $T^C$;
\item $W\setminus \{x\}$ contains a vertex with degree at least $3$ in $T^C$;
\item we have
\[
\sum_{w\in W\setminus \{x\}} \big(|\delta_{F\cup \OPT}(v)| -2\big)\ \ge\  2-q
\]
where $q$ is the number of links $\ell\in P^{xz} \cap S$.
\end{itemize}
\end{lemma}
\begin{proof}

Suppose that $z$ is not a leaf of $T^C$ and $W\setminus\{x\}$ contains only vertices of degree at most $2$ in $T^C$.
Then, because $W$ is the vertex set of the $x$-$z$ path in $T^C$, all vertices in $W\setminus \{x\}$ have degree exactly $2$ in $T^C$.
See Figure~\ref{fig:illustration_lower_bound}.
Let $\overline{W}$ denote the set of vertices of $T^C$ that are not contained in $W$, i.e., that are not visited by $P^{xz}$.
We observe that all vertices in $\overline{W}$ have a larger distance from $x$ in $T^C$ than $z$.
Here, we used that $x$ is a leaf of $T^C$ and all vertices on the $x$-$z$ path $P^{xz}$ in $T^C$, which has vertex set $W$, have degree $2$ in $T^C$.

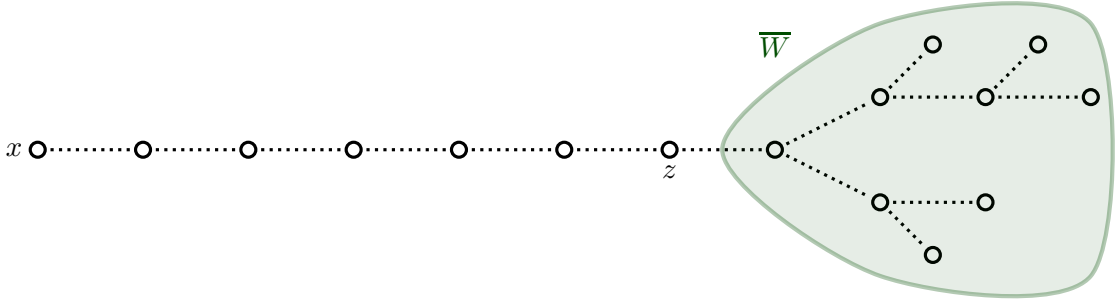
\begin{figure}[H]
\begin{center}
\begin{tikzpicture}[scale=1.4]

\tikzset{vertex/.style={fill=black,circle,minimum size=0pt, inner sep=2pt, outer sep=2pt}
}

\tikzset{block/.style={draw=black, very thick, circle,minimum size=0pt, inner sep=2pt, outer sep=2pt}
}

\begin{scope}[every node/.style={block}]
\node (1) at (1,0) {};
\node (2) at (2,0) {};
\node (3) at (3,0) {};
\node (4) at (4,0) {};
\node (5) at (5,0) {};
\node (6) at (6,0) {};
\end{scope}

\begin{scope}[every node/.style={block}]
\node (0) at (0,0) {};
\node (7) at (7,0) {};

\node (w1) at (8,0.5) {};
\node (w2) at (9,0.5) {};
\node (w3) at (10,0.5) {};
\node (w4) at (8.5,1) {};
\node (w5) at (9.5,1) {};
\node (w6) at (8,-0.5) {};
\node (w7) at (9,-0.5) {};
\node (w8) at (8.5,-1) {};
\end{scope}

\node[left=2pt] (x) at (0) {$x$};
\node[below=2pt] (z) at (6) {$z$};

\begin{scope}[very thick, dotted]
\draw (0) -- (1) -- (2) -- (3) -- (4) -- (5) -- (6) -- (7);

\draw (7) -- (w1) -- (w2) -- (w3);
\draw (w1) -- (w4);
\draw (w2) -- (w5);
\draw (7) -- (w6) -- (w7);
\draw (w6) -- (w8);
\end{scope}

\draw[ultra thick, green!30!black, fill, draw, draw opacity=0.3, fill opacity=0.1] plot [smooth cycle] coordinates { (6.5,0) (8,1.2) (10,1.2) (10,-1.2) (8,-1.2)};

\node[green!30!black] () at (7,1) {$\overline{W}$};

\end{tikzpicture}
\end{center}
\caption{Illustration of $T^C$ in the proof of Lemma~\ref{lem:opt_credits_lower_bound}.
\label{fig:illustration_lower_bound}
}
\end{figure}

Because $(V,\OPT \cup F)$ is 2-edge-connected, it contains two edge-disjoint paths $P_1,P_2$ from $x$ to $\overline{W}$.
Each of these paths visits at least one vertex in $W\setminus \{x\}$ due to the following. 
Let $i\in\{1,2\}$ and let $u$ be the first vertex on $P_i$ that is distinct from $x$ and contained in $T^C$; this vertex exists because the last vertex of $P_i$ is an element of $\overline{W}$.
Then the subpath of $P_i$ from $x$ to $u$ is an alternating trail.
By the choice of $z$, this implies $u \in W$.

For $q\ge 2$, the statement of the lemma is trivial and thus it suffices to distinguish the following two cases.
\smallskip

\noindent \textbf{Case 1:} $q=0$
\smallskip

One of the edge-disjoint paths $P_1,P_2$, say $P_1$, does not contain the unique edge incident to $x$ in $P^{xz}$. 
Let $u_1$ be the first vertex in $W$ visited by $P_1$ and let $v_1, v_2$ be the last vertices in $W$ visited by $P_1,P_2$, respectively.
See Figure~\ref{fig:case_1_lower_bound_lemma}.

\begin{figure}[h]
\begin{center}
\begin{tikzpicture}[scale=1.5]

\tikzset{vertex/.style={fill=black,circle,minimum size=0pt, inner sep=2pt, outer sep=2pt}
}

\tikzset{block/.style={draw=black, very thick, circle,minimum size=0pt, inner sep=2pt, outer sep=2pt}
}

\begin{scope}[every node/.style={block}]
\node (1) at (1,0) {};
\node (2) at (2,0) {};
\node (3) at (3,0) {};
\node (4) at (4,0) {};
\node (5) at (5,0) {};
\node (6) at (6,0) {};
\end{scope}

\begin{scope}[every node/.style={block}]
\node (0) at (0,0) {};
\node[draw=none] (7) at (7,0) {};
\node (a1) at (1.3,1) {};
\node (a2) at (2.7,1) {};
\node (b1) at (3.5,-1) {};
\node (b2) at (5,-1) {};
\node (b3) at (6.5,-1) {};
\end{scope}

\node[left=2pt] (x) at (0) {$x$};
\node[below=2pt] (z) at (6) {$z$};
\node[above=2pt] (v1) at (6) {$v_1$};
\node[above=2pt] (v2) at (2) {$v_2$};
\node[below=2pt] (u1) at (4) {$u_1$};

\begin{scope}[very thick]
\draw[red] (0) -- (a1) -- (a2) -- (4);
\draw[blue] (2) -- (b1) -- (b2) -- (b3) -- (7.7,-0.5);
\end{scope}

\begin{scope}[very thick, densely dashed]
\draw (2) -- (3) -- (4);
\draw[red] (4) -- (5) -- (6);
\draw[red, dotted] (6) -- (7);
\draw[blue] (0) -- (1) -- (2);
\end{scope}

\draw[ultra thick, green!30!black, fill, draw, draw opacity=0.3, fill opacity=0.1] plot [smooth cycle] coordinates { (6.5,0) (8,1.2) (10,1.2) (10,-1.2) (8,-1.2)};

\node[green!30!black] () at (7,1) {$\overline{W}$};
\node[red] (p1) at (3.5,0.8) {$P_1$};
\node[blue] (p2) at (2.5,-0.6) {$P_2$};

\end{tikzpicture}
\end{center}
\caption{Example of Case 1.
Dashed edges are part of the forest $F$, while solid edges are links, i.e., elements of $L$. Dotted edges can be any edges in $F\cup L$.
\label{fig:case_1_lower_bound_lemma}
}
\end{figure}
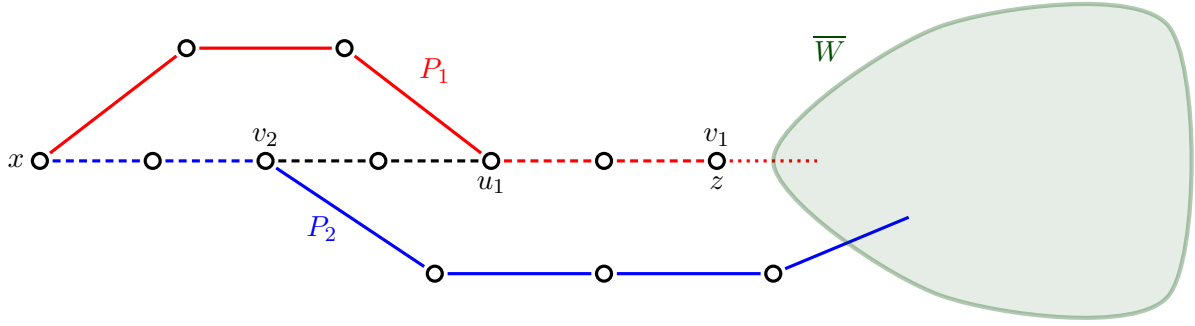

Then the edge $e_0$ by which $P_1$ enters $u_1$ is not contained in $P^{xz}$; here we used that $P_1$ does not contain the unique edge incident to $x$ in $P^{xz}$. 
Moreover, because $P_1$ and $P_2$ end in $\overline{W}$, for each $i\in\{1,2\}$ the vertex $v_i \in W$ is not the endpoint of $P_i$ and thus there is an edge $e_i\in P_i$ by which $P_i$ leaves $v_i$. 
By the definition of $v_i$, this edge is not contained in $P^{xz}$.
Using that $q=0$ implies $P^{xz}\subseteq F$, we conclude
\[
\sum_{v\in W\setminus \{x\}} \big(|\delta_{F\cup \OPT}(v)| -2\big) \ \ge\ 3 + \sum_{v\in W\setminus \{x\}} \big(|\delta_{P^{xz}}(v)| -2\big) \ge 2 \ =\ 2-q.
\]

\noindent \noindent \textbf{Case 2:} $q=1$
\smallskip

Let $\{u,v\}$ be the unique link in $P^{xz}\cap S$, where without loss of generality $u$ is closer to $x$ in $P^{xz}$.
We assume $|\delta_{F\cup \OPT}(v')|=2$ for all $v'\in W$ and derive a contradiction.
Because $P^{xz} \setminus \{u,v\} \subseteq F$, this assumption implies that if a path $P_i$ (with $i\in \{1,2\}$) visits the vertex $u$, then the $x$-$u$ subpath of $P^{xz}$ is contained in $P_i$.
Similarly, if a path $P_i$ (with $i\in \{1,2\}$) visits $v$, then the $v$-$z$ subpath of $P^{xz}$ is contained in $P_i$.
Moreover, the assumption also implies that every vertex $w\in W\setminus \{x,z\}$ is visited by at most one of the paths $P_1,P_2$.
Combining these observations with the fact that each of the paths $P_1,P_2$ visits at least one vertex in $W$, we obtain the following.
One of the paths $P_1,P_2$, say $P_1$ starts with the $x$-$u$ subpath of $P^{xz}$ and then visits no further vertices in $W\setminus \{z\}$;
the other path, $P_2$, contains the $v$-$z$ subpath of $P^{xz}$ and does not visit any vertex in $W\setminus \{x\}$ before $v$.
See Figure~\ref{fig:case_2_lower_bound_lemma}.

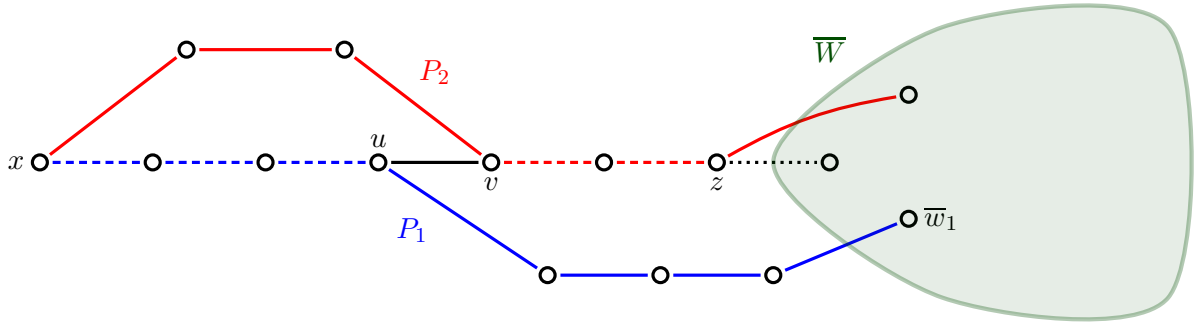
\begin{figure}[h]
\begin{center}
\begin{tikzpicture}[scale=1.5]

\tikzset{vertex/.style={fill=black,circle,minimum size=0pt, inner sep=2pt, outer sep=2pt}
}

\tikzset{block/.style={draw=black, very thick, circle,minimum size=0pt, inner sep=2pt, outer sep=2pt}
}

\begin{scope}[every node/.style={block}]
\node (1) at (1,0) {};
\node (2) at (2,0) {};
\node (3) at (3,0) {};
\node (4) at (4,0) {};
\node (5) at (5,0) {};
\node (6) at (6,0) {};
\end{scope}

\begin{scope}[every node/.style={block}]
\node (0) at (0,0) {};
\node (7) at (7,0) {};
\node (w) at (7.7,-0.5) {};
\node (w2) at (7.7,0.6) {};

\node (a1) at (1.3,1) {};
\node (a2) at (2.7,1) {};
\node (b1) at (4.5,-1) {};
\node (b2) at (5.5,-1) {};
\node (b3) at (6.5,-1) {};
\end{scope}

\node[left=2pt] (x) at (0) {$x$};
\node[below=2pt] (z) at (6) {$z$};

\node[above=2pt] (v2) at (3) {$u$};
\node[below=2pt] (u1) at (4) {$v$};
\node[right=2pt] () at (w) {$\overline{w}_1$};

\begin{scope}[very thick]
\draw[red] (0) -- (a1) -- (a2) -- (4);
\draw[red, bend left=10] (6) to (w2);
\draw[blue] (3) -- (b1) -- (b2) -- (b3) -- (w);
\draw (3) -- (4);
\end{scope}

\begin{scope}[very thick, densely dashed]
\draw[red] (4) -- (5) -- (6);
\draw[dotted] (6) -- (7);
\draw[blue] (0) -- (1) -- (2) -- (3);
\end{scope}

\draw[ultra thick, green!30!black, fill, draw, draw opacity=0.3, fill opacity=0.1] plot [smooth cycle] coordinates { (6.5,0) (8,1.2) (10,1.2) (10,-1.2) (8,-1.2)};

\node[green!30!black] () at (7,1) {$\overline{W}$};
\node[red] (p2) at (3.5,0.8) {$P_2$};
\node[blue] (p1) at (3.3,-0.6) {$P_1$};

\end{tikzpicture}
\end{center}
\caption{Example illustrating the proof of Case 2.
Dashed edges are part of the forest $F$, while solid edges are links, i.e., elements of $L$. Dotted edges can be any edges in $F\cup L$.
\label{fig:case_2_lower_bound_lemma}
}
\end{figure}

If the $x$-$v$ subpath of $P_2$ visits a vertex $\overline{w}\in\overline{W}$, then this contradicts the choice of $z$ (because if we choose $\overline{w}$ to be the first vertex on $P_2$  that is contained in $\overline{W}$, the $x$-$\overline{w}$ subpath of $P_2$ would be an alternating trail).

Let $\overline{w}_1$ be the first vertex on $P_1$ that is contained in $\overline{W}$.
(Such a vertex exists because $P_1$ ends in $\overline{W}$.)
Next, we show that there exists an alternating $x$-$\overline{w}_1$ trail, leading to a contradiction.
If the $x$-$v$ subpath of $P_2$ and the $u$-$\overline{w}_1$ subpath of $P_1$ are not vertex disjoint, then there exists a $x$-$\overline{w}_1$ path that does not contain any vertex of $T^C$ as an interior vertex; this path is an alternating  $x$-$\overline{w}_1$ trail.
Otherwise, we obtain an alternating  $x$-$\overline{w}_1$ trail by concatenating the $x$-$v$ subpath of $P_2$, the link $\{v,u\}$, and the $u$-$\overline{w}_1$ subpath of $P_1$.
\end{proof}

By Lemma~\ref{lem:new_component_bridge_covering}, every vertex whose degree increases in the augmentation step becomes part of a 2EC-block and is not lonely.
Thus, Invariant~\ref{invariant:lonely_vertices} is maintained.
We now show that also Invariant~\ref{invariant:credits} is maintained. 

\begin{lemma}
Let $z$ be the vertex that is furthest away from $x$ in $T^C$ among all vertices reachable from $x$ by an alternating trail and let $P$ be an alternating $x$-$z$ trail.
When augmenting along $P$, the number of bridges in $H$ reduces by at least $1$.
Moreover, $\cred(H) +|S|$ does not increase.
\end{lemma}
\begin{proof}
By Lemma~\ref{lem:new_component_bridge_covering}, when augmenting along $P$, we do not create any new bridges and all of the edges on the path $P^{xz}$ become part of a 2-edge-connected component.
This implies that the number of bridges of $H$ decreases by at least $1$.

Let $k$ denote the number of links in $P\cap P^{xz}$.
When augmenting along $P$, the cardinality of $S$ increases by exactly $|P|-2k$.
Moreover, the number of connected components of $H$ decreases by exactly $|P|-1-2k$.
We distinguish two cases.\\

\noindent
\textbf{Case 1:} After the augmentation along $P$, the connected component containing $x$ contains at least one bridge. \\

Each of the $|P|-2k$ connected components that get merged when augmenting along $P$ (including the component $C$) had at least $1$ credit before the augmentation by \ref{item:component_credits}.
The connected component resulting from the augmentation has $1$ credit and thus the number of credits due to \ref{item:component_credits} decreases by at least 
$|P|-1-2k$.
Therefore, it remains to show that the number of credits due to \ref{item:vertex_credits}, \ref{item:bridge_credits}, and \ref{item:2ec_component_credits} reduces by at least $1$ when augmenting along $P$.

By Lemma~\ref{lem:new_component_bridge_covering}, when augmenting along $P$, all vertices of the path $P^{xz}$ become part of the same 2EC-block $B$ and we do not create any new bridges in $H$.
The block $B$ has exactly $1$ credit by \ref{item:2ec_component_credits}.
The credits of all vertices and blocks that are not contained in $B$ remain unchanged.
Moreover, the credits of links in $S\setminus P^{xz}$ do not change.
Vertices and links in $B$ have no credits anymore after the augmentation.
Hence, it suffices to show that the total number of credits on vertices, bridges, and 2EC-blocks in $P^{xz}$ was at least $2$ before the augmentation.

Before the augmentation, $x$ was a leaf of $T^C$ and thus it was either a leaf of $H$ or it resulted from the contraction of a 2EC-block. 
In both cases it had at least $1$ credit by \ref{item:leaf_credits} and \ref{item:2ec_component_credits}.

Let $W$ denote the vertex set of $P^{xz}$.
If $W\setminus \{x\}$ contains a vertex that arose from the contraction of a 2EC-block, this block had $1$ credit before the augmentation by \ref{item:2ec_component_credits}.
Otherwise, all vertices in $W\setminus\{x\}$ are lonely and hence they have degree at most $2$ by Invariant~\ref{invariant:lonely_vertices}.
If $z$ is a leaf of $T^C$ it has $1$ credit by \ref{item:leaf_credits} before the augmentation.
Otherwise, Lemma~\ref{lem:opt_credits_lower_bound} implies that before the augmentation the total number of credits on bridges in $P^{xz}$ (due to \ref{item:bridge_credits}) plus the total number of credits on $W\setminus \{x\}$ (due to \ref{item:lonely_credits}) was at least $1$.
(Here, we used that all links in $P^{xz}$ were bridges before the augmentation and that $1-\epsilon \ge \tfrac{1}{2}$.)\\

\noindent
\textbf{Case 2:} After the augmentation along $P$, the connected component containing $x$ is 2-edge-connected.
\\

The connected component $C$ had $1$ credit before the augmentation.
After the augmentation, the connected component containing $x$ has at most $2$ credits.
The number of all other types of credits does not increase.
Moreover, because the connected component containing $x$ is 2-edge-connected after the augmentation, both $x$ and $z$ must be leaves of the tree $T^C$.
Thus, $x$ and $z$ were either leaves of $H$ or (contracted) 2EC-blocks of $H$.
This implies that the total number of credits according to \ref{item:leaf_credits} and \ref{item:2ec_component_credits} decreases by at least $2$.
\end{proof}

\subsection{Gluing}\label{sec:gluing}

In this section we describe and analyze the gluing phase. 
For the purpose of gluing, we consider auxiliary graphs that are obtained from the graph $(V,F\cup L)$ by contracting some connected components of the current $H=(V,F\cup S)$, similarly to the bridge covering phase.
We will merge several connected components of $H$ by finding an affordable cycle in one of the auxiliary graphs. We begin by defining these auxiliary graphs, and our main tools, namely good cycles.

Recall that a graph $G$ is called bridgeless if it does not contain any bridge (i.e.\ cut-edge). 
Let $F\cup S$ be a bridgeless $2$-edge-cover of $(V,F\cup L)$ and $H=(V,F\cup S)$. 
Throughout the gluing phase $F\cup S$ will remain bridgeless and thus the connected components and the 2-edge-connected components of $H$ coincide. Hence, we will sometimes refer to these as the components of $H$.
We let $G_H$ be the (multi-)graph obtained from $(V,F\cup L)$ after contracting each component of $H$ into a single vertex. 
Notice that the edges of $G_H$ correspond to a subset of the links of the original graph. 

For a simple component $C$, we let $G_{H|C}$ be the (multi-)graph obtained from $(V,F\cup L)$ after contracting each connected components of $H$ except for $C$ into a single vertex. 
Then the edges of $G_{H|C}$ correspond to a subset of the links of the original graph plus the two paths $P_1 \subseteq F$ and $P_2 \subseteq F$ in the simple component $C$.
Note that $G_H$ and $G_{H|C}$ are not necessarily simple graphs and might have parallel edges.
With a slight abuse of notation we identify the edges in $G_H$ and $G_{H|C}$ with the edges in $(V,F\cup L)$.

\begin{definition}[good cycle]
A \emph{good cycle of $H$ in $(V,F\cup L)$ that affects a simple component $C$} is a cycle in $G_{H|C}$ that contains all the edges of the two paths $P_1$ and $P_2$ of $C$ in $F$ and has at least one more vertex. 
\end{definition}

See Figure~\ref{fig:GoodCycle}.
In the gluing phase, we will repeatedly merge components of $H$ by \emph{gluing} either along a cycle in $G_H$ or along a good cycle in $G_{H|C}$ (for some simple component $C$).
Next, we define these gluing operations.

\begin{definition}\label{def:gluing}
\emph{Gluing $H$ along a cycle $Q$ of $G_H$} is the process of adding the links of $Q$ to $H$. 

\emph{Gluing $H=(V,F\cup S)$ along a good cycle $Q$ affecting a simple component $C$} is the process of first removing the two links of $C$ from $S$ and then adding the links of $Q$ to $S$.
\end{definition}

Notice that it is possible that one link $\ell$ belonging to both $C$ and $Q$  is first removed and then added back. However this cannot happen for both links in the simple component $C$ simultaneously. See Figure~\ref{fig:GoodCycle}.
We also remark that in order to obtain a better-than-2 approximation it is not sufficient to use only gluing operations that glue $H$ along cycles in $G_H$ as the example from Figure~\ref{fig:removing_needed_gluing} shows.

\begin{figure}[H]
\centering
\begin{subfigure}[c]{0.3\textwidth}
\begin{center}
\begin{tikzpicture}[scale=0.7]
\useasboundingbox (-1.5, -0.5) rectangle (5.5, 5.5);

\tikzset{vertex/.style={fill=black,circle,minimum size=0pt, inner sep=2pt, outer sep=2pt}
}

\tikzset{comp/.style={draw=black, very thick, circle,minimum size=0pt, inner sep=2pt, outer sep=2pt}
}

\begin{scope}[every node/.style={vertex}]
\node (v2) at (0,0) {};
\node (v1) at (0,4) {};
\node (u2) at (4,0) {};
\node (u1) at (4,4) {};

\node (w1) at (4,1.3) {};
\node (w2) at (4,2.7) {};
\end{scope}

\node[left=2pt] (V1) at (v1) {$v_1$};
\node[left=2pt] (V2) at (v2) {$u_1$};

\node[right=2pt] (U1) at (u1) {$v_2$};
\node[right=2pt] (U2) at (u2) {$u_2$};

\node (P1) at (-0.35,2) {$P_1$};
\node (P2) at (4.4,2) {$P_2$};

\begin{scope}[very thick,densely dashed]
\draw (v1) -- (v2);
\draw (u1) -- (w2) -- (w1) -- (u2);
\end{scope}

\begin{scope}[very thick]
\draw (v1) -- (u1);
\draw (v2) -- (u2);
\end{scope}

\end{tikzpicture}
\end{center}
\caption{a simple component \label{subfig:simple_comp}}
\end{subfigure}
\hfill
\begin{subfigure}[c]{0.3\textwidth}
\begin{center}
\begin{tikzpicture}[scale=0.7]
\useasboundingbox (-1.5, -0.5) rectangle (5.5, 5.5);

\tikzset{vertex/.style={fill=black,circle,minimum size=0pt, inner sep=2pt, outer sep=2pt}
}

\tikzset{comp/.style={draw=black, very thick, circle,minimum size=0pt, inner sep=2pt, outer sep=2pt}
}

\begin{scope}[every node/.style={vertex}]
\node (v2) at (0,0) {};
\node (v1) at (0,4) {};
\node (u2) at (4,0) {};
\node (u1) at (4,4) {};

\node (w1) at (4,1.3) {};
\node (w2) at (4,2.7) {};
\end{scope}

\node[left=2pt] (V1) at (v1) {$v_1$};
\node[left=2pt] (V2) at (v2) {$u_1$};

\node[right=2pt] (U1) at (u1) {$v_2$};
\node[right=2pt] (U2) at (u2) {$u_2$};

\begin{scope}[every node/.style={comp}]
\node (C1) at (3,5) {};
\node (C2) at (1,5) {};
\end{scope}

\begin{scope}[line width=1.5pt,densely dashed]
\draw[blue] (v1) -- (v2);
\draw[blue] (u1) -- (w2) -- (w1) -- (u2);
\end{scope}

\begin{scope}[very thick]
\draw (v1) -- (u1);
\draw[line width=1.5pt,blue] (v2) -- (u2);
\end{scope}

\draw[line width=1.5pt,blue] (u1) -- (C1) -- (C2) -- (v1);

\end{tikzpicture}
\end{center}
\caption{a good cycle \label{subfig:good_cycle1}}
\end{subfigure}
\hfill
\begin{subfigure}[c]{0.3\textwidth}
\begin{center}
\begin{tikzpicture}[scale=0.7]
\useasboundingbox (-1.5, -0.5) rectangle (5.5, 5.5);

\tikzset{vertex/.style={fill=black,circle,minimum size=0pt, inner sep=2pt, outer sep=2pt}
}

\tikzset{comp/.style={draw=black, very thick, circle,minimum size=0pt, inner sep=2pt, outer sep=2pt}
}

\begin{scope}[every node/.style={vertex}]
\node (v2) at (0,0) {};
\node (v1) at (0,4) {};
\node (u2) at (4,0) {};
\node (u1) at (4,4) {};

\node (w1) at (4,1.3) {};
\node (w2) at (4,2.7) {};
\end{scope}

\node[left=2pt] (V1) at (v1) {$v_1$};
\node[left=2pt] (V2) at (v2) {$u_1$};

\node[right=2pt] (U1) at (u1) {$v_2$};
\node[right=2pt] (U2) at (u2) {$u_2$};

\begin{scope}[every node/.style={comp}]
\node (C1) at (1.3,2.7) {};
\node (C2) at (2.7,1.3) {};
\node (D1) at (3.5,-1) {};
\node (D2) at (5.3,0.5) {};
\end{scope}

\begin{scope}[line width=1.5pt,densely dashed]
\draw[blue] (v1) -- (v2);
\draw[blue] (u1) -- (w2) -- (w1) -- (u2);
\end{scope}

\begin{scope}[very thick]
\draw (v1) -- (u1);
\draw (v2) -- (u2);
\end{scope}

\draw[line width=1.5pt, blue] (v1) -- (C1) -- (C2) -- (u2);
\draw[line width=1.5pt,blue] (v2) to[out=320, in=180] (D1) to[out=5, in=260] (D2) to[out=90, in=320] (u1);

\end{tikzpicture}
\end{center}
\caption{another good cycle \label{subfig:good_cycle2}}
\end{subfigure}

\caption{
This figure shows a simple component $C$ (a), and two different good cycles affecting $C$ (shown in blue in (b) and (c)).
Solid edges are links and dashed edges are edges of the forest $(V,F)$.
Filled vertices show elements of the original vertex set $V$, while other vertices resulted from the contraction of a component of $H$.
Notice that when gluing along the good cycle depicted in (b), the link $\{u_1,u_2\}$ will remain in $S$.
(Strictly speaking, according to Definition~\ref{def:gluing}, we first remove it, but then add it again.)
When gluing $H$ along the good cycle depicted in (c), both links $\{v_1,v_2\}$ and $\{u_1,u_2\}$ will be removed from $S$.
}\label{fig:GoodCycle}
\end{figure}
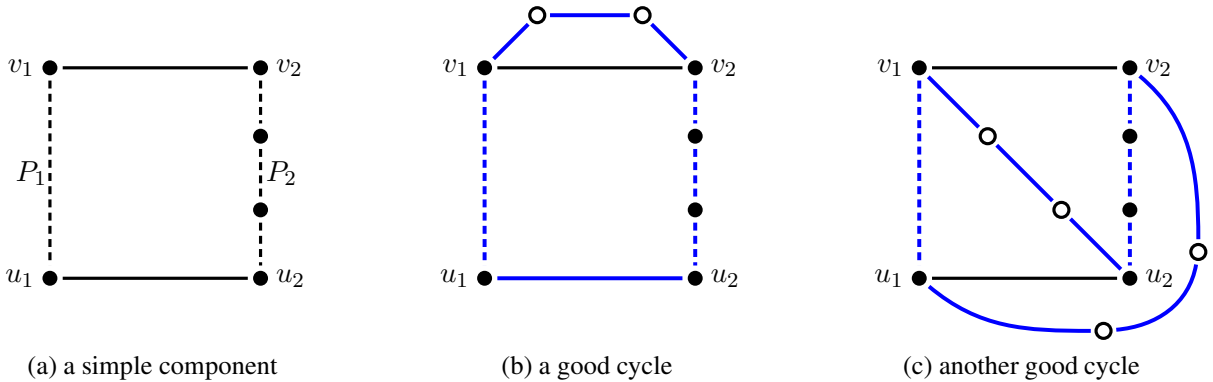

\begin{lemma}\label{lem:ComputationOfGoodCycles}
Given a bridgeless $2$-edge-cover $F\cup S$ of $(V,F\cup L)$, there exists a polynomial time algorithm that computes a good cycle of $H=(V,F\cup S)$ in $(V,F\cup L)$ or verifies that no such good cycle exists. 
\end{lemma}

\begin{proof}
Fix a simple component $C$ of $H$. We show that there exists a polynomial time algorithm that computes a good cycles of $H$ in $(V,F\cup L)$ that affects $C$ or verifies that no such good cycle exist. The claim follows by applying this algorithm to each simple component of $H$.

Let $P_1$ and $P_2$ be the paths of $C$ that belong to $F$ and let $u_1$ and $v_1$ be the endpoints of $P_1$ and $u_2$ and $v_2$ be the endpoints of $P_2$ such that $\{u_1,u_2\}\in S$ and $\{v_1,v_2\}\in S$ are the unique links of the simple component $C$. See Figure~\ref{fig:GoodCycle}.
We call a path in $G_{H|C}$ a \emph{good path} if one of its endpoints is $u_1$ or $v_1$, the other endpoint is $u_2$ or $v_2$, and all its internal vertices do not belong to $C$, i.e., they are vertices that arose from the contraction of components of $H$. 
Every good cycle affecting $C$ is the union of $P_1$, $P_2$ and two vertex disjoint good paths. 
More precisely, if there exists a good cycle affecting $C$ then
\begin{enumerate}[label=(\alph*)]\itemsep0pt
\item there exists a good $v_1$-$v_2$ path or a good $u_1$-$u_2$ path of length at least two, or \label{item:good_cycle_simple}
\item there exist a good $v_1$-$u_2$ path and a good $u_1$-$v_2$ path such that at least one of these paths has length at least two.
\label{item:good_cycle_complicated}
\end{enumerate}
We give a polynomial time algorithm that finds a good cycles whenever  \ref{item:good_cycle_simple} of \ref{item:good_cycle_complicated} holds.
First we observe that for vertices $a\in \{u_1,v_1\}$ and $b\in\{u_2,v_2\}$ we can in polynomial time find a good $a$-$b$ path or decide that such a path does not exist (by searching for an $a$-$b$ path in the graph arising from $G_{H|C}$ by removing all vertices of $C$ (and their incident edges) except for the vertices $a$ and $b$).
We can also find a good $a$-$b$ path of length at least two or decide that no such path exists by looking for a good $a$-$b$ path in the graph where we remove the link $\{a,b\}$ (if it exists).

Our algorithm first checks if a good $v_1$-$v_2$ path $P_v$ of length at least two exists. If this is the case, we obtain a good cycle by taking the union of $P_1$, $P_v$, $P_2$, and the link $\{u_1,u_2\}$.
We proceed analogously if a good $u_1$-$u_2$ path of length at least two exists.
Thus, we find a good cycle if \ref{item:good_cycle_simple} holds.
Otherwise, we check if there exists a good $v_1$-$u_2$ path $P_1'$ and a good $v_2$-$u_1$ path $P_2'$.
If possible, we choose the good paths $P_1'$ and $P_2'$ such that they have length at least two.
If \ref{item:good_cycle_complicated} holds, both $P_1'$ and $P_2'$ exist and at least one of them has length at least two. 
The paths $P_1'$ and $P_2'$ must be vertex-disjoint as otherwise their union contains a good $v_1$-$v_2$ path of length at least two, in which case \ref{item:good_cycle_simple} holds.
Then the union of $P_1$, $P_1'$, $P_2$ and $P_2'$ is a good cycle.
\end{proof}

Next we show that if we can glue $H$ along a good cycle, our credit invariant (Invariant~\ref{invariant:credits}) is maintained.
We also observe that the number of components of $H$ decreases strictly.

\begin{lemma}\label{lem:goodCycleCost}
Given a bridgeless $2$-edge-cover $F\cup S$ of $(V,F\cup L)$ and a good cycle $Q$ of $H=(V,F\cup S)$ in $(V,F\cup L)$, let $H'=(V,F\cup S')$ be the graph obtained by gluing $H$ along $Q$. 
Then $F\cup S'$ is a bridgeless $2$ edge-cover of $(V,F\cup L)$ such that $\cred(H)+|S|\ge \cred(H')+|S'|$ and $H'$ has fewer connected components than $H$. 
\end{lemma}

\begin{proof}
Let $C$ be the simple component of $H$ that is affected by $Q$ and let  $C_1,...,C_k$ be the other connected components of $H$ such that their corresponding vertex in $G_{H|C}$ belongs to $Q$.

Clearly $H'$ is also a bridgeless $2$-edge-cover. 
Furthermore, the vertices of $C$ and the vertices of $C_1,\dots,C_k$, form a single connected component $C'$ in $H'$. Observe that we have removed $2$ links and added $k+2$ links, therefore $|S'|-|S|=k$. 
By the credit rule \ref{item:component_credits}, we need $2$ credits for the new component $C'$ (which is not simple since it contains at least $3$ links). 
We can compensate this with the credits of the components $C$, $C_1, \dots, C_k$ of $H$, which are at least $(k+1)(2-2\eps)\ge k+2$ (since $\eps=\frac{1}{4}$ and $k\geq 1$). 
Thus, $\cred(H)-\cred(H')\ge k$ and therefore
 $\cred(H)+|S|-\cred(H')-|S'|\ge k-k = 0$.

Finally, we observe that the connected components of $H$ not touched by $Q$ are not modified, hence their credits do not change.
\end{proof}

It remains to consider the case where there is no good cycle.
In order to prove the credit invariant, we first show that in this case every simple component contains a vertex with an implicit credit due to \ref{item:lonely_credits}.
Such vertices will be called \emph{rich}.

\begin{definition}[rich vertex]
We say that a vertex $v\in V$ is \emph{rich} if ~$|\delta_{\OPT \cup F}(v)|\ge 3$ and it belongs to a simple connected component of $H$.
\end{definition}
\begin{lemma}\label{lemma:GoodCycleOrBadVertex}
Let $C$ be a simple component of $H=(V,F\cup S)$ and suppose $C\ne H$. 
If there is no good cycle affecting $C$, then $C$ contains a rich vertex.
\end{lemma}
\begin{proof}
Let $P_1\subseteq F$ and $P_2\subseteq F$ be the paths of $F$ in $C$, such that $v_1$, $u_1$ are the endpoints of $P_1$, $v_2$ and $u_2$ are the endpoints of $P_2$, and $\{v_1, v_2\}, \{u_1,u_2\}\in S$ are the two links contained in $C$.

Assume that the claim does not hold, i.e., $C$ does not have any rich vertex. Then in $(V,F\cup \OPT)$ the degree of each vertex of $C$ is exactly $2$. 
Therefore, the internal vertices of $P_1$ and $P_2$ are not incident to any edge of $\OPT$ and each of their endpoints $u_1$, $u_2$, $v_1$, and $v_2$ is incident to exactly one edge of $\OPT$. 

Now consider the graph $G_{H|C}$.
Since $(V,F\cup \OPT)$ is $2$-edge-connected, there must exist two edge-disjoint paths $P'_1$ and $P'_2$ from the vertices of $P_1$ to the vertices of $P_2$ in $G_{H|C}$ that only contain links of $\OPT$.
 Observe that, since by assumption $C$ has no rich vertex, the set of endpoints of $P'_1$ and $P'_2$ is exactly $\{u_1,u_2,v_1,v_2\}$ and none of the internal vertices of $P'_1$ and $P'_2$ belongs to $C$.

By  a case analysis similar to the one in the proof of Lemma \ref{lem:ComputationOfGoodCycles}, the only possibility for $P'_1\cup P'_2$ not to induce a good cycle affecting $C$ is that $P'_1\cup P'_2$ is a matching of size $2$ between the endpoints of $P_1$ and the endpoints of $P_2$.
However $\OPT$ must contain two links between the vertex set $V(C)$ of $C$ and its complement $V\setminus V(C)$ (since $C\neq H$ and $(V,F\cup \OPT)$ is $2$-edge-connected).
Hence, in this case some endpoint of $P_1$ or of $P_2$ must be rich, a contradiction.
\end{proof}

Finally, we complete the analysis of the gluing phase by showing that the credit invariant (Invariant~\ref{invariant:credits}) is maintained and the runtime is polynomial.

\begin{lemma}
Given a bridgeless $2$-edge-cover $H$ of $(V,F\cup L)$, the gluing phase of Algorithm~\ref{algo:pap_overview} yields a $2$-edge-connected spanning subgraph $H^{*}=(V,F\cup S^*)$ of $(V,F\cup L)$ such that $|S|+ \cred(H)\ge |S^{*}|+\cred(H^{*})$ in polynomial time.
\end{lemma}
\begin{proof}
For this purpose, we show that if $H$ is not already $2$-edge-connected, then in every iteration of the gluing phase, we obtain in polynomial time a bridgeless subgraph $H'=(V,F\cup S')$ of $(V,F\cup L)$ from $H=(V,F\cup S)$ such that $|S|+\cred(H)\ge |S'|+\cred(H')$ and $H'$ has fewer connected component than $H$.
This implies that the   number of iterations is linear in $|V|$.

First observe that if $H$ has a good cycle, then using Lemmas~\ref{lem:ComputationOfGoodCycles} and \ref{lem:goodCycleCost} we are done. 
Therefore we can assume that no good cycle of $H$ exists in $(V,F\cup L)$.
 Then, by Lemma~\ref{lemma:GoodCycleOrBadVertex} every simple component of $H$ has at least one rich vertex, which has at least $\tfrac{1}{2}$ credits by \ref{item:lonely_credits}.
Thus, every component of $H$ is either simple with at least one rich vertex or is non-simple. In particular the credits of each simple component $C$ plus the credits of its rich vertices add up to at least $(2-2\eps)+\frac{1}{2} =  2$. 

It remains to show that gluing $H=(V,F\cup S)$ along an arbitrary cycle $Q=(W,L')$ in $G_H$ (such a cycle must exist because $(V,F\cup L)$ is $2$-edge-connected) does not increase $\cred(H)+|S|$.
Let $H'=(V,F\cup S')$ be the resulting graph, where $S'=S\cup L'$ and $C'$ is the new connected component formed during the gluing. 
Notice that $C'$ is not simple since it contains strictly more than two links.
Observe that $|S'|-|S|=|L'|$, and we need $2$ credits for $C'$ (due to the credit rule \ref{item:component_credits}).
From the above discussion each of the $|L'|$ many connected component of $H$ involved in $Q$ can contribute with at least $2$ credits. Thus, 
\[
|S|+\cred(H)-|S'|-\cred(H')\geq -|L'|+2 \cdot |L'|- 2 \geq 0. \qedhere
\]
\end{proof}

\bibliographystyle{alpha}%

\begin{thebibliography}{CDG{\etalchar{+}}20}

\bibitem[Adj17]{A17}
David Adjiashvili.
\newblock Beating approximation factor two for weighted tree augmentation with
  bounded costs.
\newblock In {\em ACM-SIAM Symposium on Discrete Algorithms (SODA)}, pages
  2384--2399, 2017.

\bibitem[BGJA20]{BGJ20}
Jaroslaw Byrka, Fabrizio Grandoni, and Afrouz Jabal-Ameli.
\newblock Breaching the 2-approximation barrier for connectivity augmentation:
  a reduction to {Steiner} tree.
\newblock In {\em {ACM} {SIGACT} Symposium on Theory of Computing (STOC)},
  pages 815--825, 2020.

\bibitem[BGRS13]{BGRS13}
Jaroslaw Byrka, Fabrizio Grandoni, Thomas Rothvo{\ss}, and Laura Sanit{\`{a}}.
\newblock {Steiner Tree} approximation via iterative randomized rounding.
\newblock {\em J. {ACM}}, 60(1):6:1--6:33, 2013.

\bibitem[CCDZ20]{CCDZ20}
Joseph Cheriyan, Robert Cummings, Jack Dippel, and Jasper Zhu.
\newblock An improved approximation algorithm for the matching augmentation
  problem.
\newblock {\em CoRR}, abs/2007.11559, 2020.

\bibitem[CDG{\etalchar{+}}20]{CDGKN20}
Joseph Cheriyan, Jack Dippel, Fabrizio Grandoni, Arindam Khan, and Vishnu~V.
  Narayan.
\newblock The matching augmentation problem: a $7/4$-approximation algorithm.
\newblock {\em Math. Program.}, 182(1):315--354, 2020.

\bibitem[CN13]{CN13}
Nachshon Cohen and Zeev Nutov.
\newblock A (1+ln 2)-approximation algorithm for minimum-cost
  2-edge-connectivity augmentation of trees with constant radius.
\newblock {\em Theor. Comput. Sci.}, 489-490:67--74, 2013.

\bibitem[CSS01]{CSS01}
Joseph Cheriyan, Andr{\'{a}}s Seb{\"{o}}, and Zolt{\'{a}}n Szigeti.
\newblock Improving on the 1.5-approximation of a smallest 2-edge connected
  spanning subgraph.
\newblock {\em {SIAM} J. Discret. Math.}, 14(2):170--180, 2001.

\bibitem[CT00]{CT00}
Joseph Cheriyan and Ramakrishna Thurimella.
\newblock Approximating minimum-size k-connected spanning subgraphs via
  matching.
\newblock {\em {SIAM} J. Comput.}, 30(2):528--560, 2000.

\bibitem[CTZ21]{CTZ21}
Federica Cecchetto, Vera Traub, and Rico Zenklusen.
\newblock Bridging the gap between tree and connectivity augmentation: unified
  and stronger approaches.
\newblock In {\em {ACM} {SIGACT} Symposium on Theory of Computing (STOC)},
  pages 370--383, 2021.

\bibitem[DKL76]{DKL76}
Evim~A. Dinits, Alexander~V. Karzanov, and Michael~V. Lomonosov.
\newblock On the structure of a family of minimal weighted cuts in a graph.
\newblock {\em Studies in Discrete Optimization}, pages 290--306, 1976.

\bibitem[EFKN09]{EFKN09}
Guy Even, Jon Feldman, Guy Kortsarz, and Zeev Nutov.
\newblock A 1.8 approximation algorithm for augmenting edge-connectivity of a
  graph from 1 to 2.
\newblock {\em ACM Transactions on Algorithms}, 5(2):21:1--21:17, 2009.

\bibitem[FGKS18]{FGKS18}
Samuel Fiorini, Martin Gro{\ss}, Jochen K{\"o}nemann, and Laura Sanit{\`a}.
\newblock Approximating weighted tree augmentation via {C}hvatal-{G}omory cuts.
\newblock In {\em ACM-SIAM Symposium on Discrete Algorithms (SODA)}, pages
  817--831, 2018.

\bibitem[GG12]{GG12}
Harold~N. Gabow and Suzanne Gallagher.
\newblock Iterated rounding algorithms for the smallest \emph{k}-edge connected
  spanning subgraph.
\newblock {\em {SIAM} J. Comput.}, 41(1):61--103, 2012.

\bibitem[GGAS19]{GGJS19}
Waldo Galvez, Fabrizio Grandoni, Afrouz~Jabal Ameli, and Krzysztof Sornat.
\newblock On the cycle augmentation problem: Hardness and approximation
  algorithms.
\newblock In {\em Workshop on Approximation and Online Algorithms (WAOA)},
  2019.

\bibitem[GKZ18]{GKZ18}
Fabrizio Grandoni, Christos Kalaitzis, and Rico Zenklusen.
\newblock Improved approximation for tree augmentation: saving by rewiring.
\newblock In {\em {ACM} {SIGACT} Symposium on Theory of Computing (STOC)},
  pages 632--645, 2018.

\bibitem[HVV19]{HVV19}
Christoph Hunkenschr{\"{o}}der, Santosh~S. Vempala, and Adrian Vetta.
\newblock A 4/3-approximation algorithm for the minimum 2-edge connected
  subgraph problem.
\newblock {\em {ACM} Trans. Algorithms}, 15(4):55:1--55:28, 2019.

\bibitem[Jai01]{J01}
Kamal Jain.
\newblock A factor 2 approximation algorithm for the generalized {Steiner}
  network problem.
\newblock {\em Combinatorica}, 21(1):39--60, 2001.

\bibitem[KKL04]{KKL04}
Guy Kortsarz, Robert Krauthgamer, and James~R. Lee.
\newblock Hardness of approximation for vertex-connectivity network design
  problems.
\newblock {\em {SIAM} Journal on Computing}, 33(3):704--720, 2004.

\bibitem[KN16]{KN16b}
Guy Kortsarz and Zeev Nutov.
\newblock A simplified 1.5-approximation algorithm for augmenting
  edge-connectivity of a graph from 1 to 2.
\newblock {\em ACM Transactions on Algorithms}, 12(2):23:1--23:20, 2016.

\bibitem[KV94]{KV94}
Samir Khuller and Uzi Vishkin.
\newblock Biconnectivity approximations and graph carvings.
\newblock {\em J. {ACM}}, 41(2):214--235, 1994.

\bibitem[Nag03]{N03}
Hiroshi Nagamochi.
\newblock An approximation for finding a smallest $2$-edge-connected subgraph
  containing a specified spanning tree.
\newblock {\em Discrete Applied Mathematics}, 126(1):83--113, 2003.

\bibitem[Nut17]{N17}
Zeev Nutov.
\newblock On the tree augmentation problem.
\newblock In {\em European Symposium on Algorithms (ESA)}, pages 61:1--61:14,
  2017.

\bibitem[Nut21]{N20}
Zeev Nutov.
\newblock Approximation algorithms for connectivity augmentation problems.
\newblock In {\em International Computer Science Symposium in Russia}, pages
  321--338, 2021.

\bibitem[Sch03]{Schrijver}
Alexander Schrijver.
\newblock {\em Combinatorial Optimization, Polyhedra and Efficiency}.
\newblock Springer, 2003.

\bibitem[SV14]{SV14}
Andr{\'{a}}s Seb{\"{o}} and Jens Vygen.
\newblock Shorter tours by nicer ears: 7/5-approximation for the graph {TSP},
  3/2 for the path version, and 4/3 for two-edge-connected subgraphs.
\newblock {\em Combinatorica}, 34(5):597--629, 2014.

\bibitem[TZ22a]{TZ21}
Vera Traub and Rico Zenklusen.
\newblock A better-than-2 approximation for weighted tree augmentation.
\newblock In {\em 2021 IEEE Symposium on Foundations of Computer Science
  (FOCS)}, pages 1--12, 2022.

\bibitem[TZ22b]{TZ22}
Vera Traub and Rico Zenklusen.
\newblock Local search for weighted tree augmentation and {Steiner} tree.
\newblock In {\em ACM-SIAM Symposium on Discrete Algorithms (SODA)}, pages
  3253--3272, 2022.

\end{thebibliography}
\newcommand{\etalchar}[1]{$^{#1}$}

\appendix

\section{Details on WTAP: Proof of Lemma  \ref{lem:wtap}}

\begin{proof}[Proof of Lemma \ref{lem:wtap}]

As observed in \cite{CN13}, we may assume that the sets $P_{\ell}$ with $\ell \in U$ are pairwise disjoint. (This can be achieved by replacing some of the links in $U$ by shadows; see e.g., Lemma~2.1 in \cite{TZ22}.)

Then we proceed as in \cite{TZ21}, but start with our given up-link solution $U$ instead of the initial up-link solution chosen in \cite{TZ21}.
More precisely, we perform steps 2--4 of Algorithm~1 in \cite{TZ21}.
The analysis is the same as in \cite{TZ21}.
Indeed, the proof of Theorem~6 from \cite{TZ21} shows that the returned solution $\mathrm{APX}$ has weight at most
\begin{align*}
& w(\OPT) + \epsilon \cdot w(\OPT) + \ln\left(\frac{w(U) - \epsilon \cdot w(\OPT) }{w(\OPT)}\right) \cdot w(\OPT) \\
\le&\ \left(1+\ln\left(\frac{w(U)}{w(\OPT)}\right) + \epsilon\right) \cdot w(\OPT).
\end{align*}
\end{proof}

\section{From FAP to PAP}\label{sec:reductionPAP}

In this section we prove Lemma \ref{lem:reducing_to_pap_noIsolated}. This is a straightforward consequence of the following two lemmas. 

\begin{lemma}\label{lem:noIsolatedNodes}
Given a polynomial-time $(\rho, K)$-approximation for FAP for the case where there is no isolated node in the input forest $(V,F)$, there is a polynomial-time $(\rho, K)$-approximation for FAP (in the general case). 
\end{lemma}
\begin{proof}[Proof]
We consider the following approximation-preserving reduction. Given an instance $\mathcal{I}_A$ of FAP containing some isolated nodes, we construct an instance $\mathcal{I}_B$ of FAP without isolated nodes as follows. 
Initially $\mathcal{I}_B=\mathcal{I}_A$. 
Then, for each isolated node $v$, we perform the following steps. 

First, we replace $v$ with two new nodes $v_1$ and $v_2$. Second, we replace each link of type $\{u,v\}$ with two new links $\{u,v_1\}$ and $\{u,v_2\}$. Third, we add the edge $\{v_1,v_2\}$ to the forest. 
Note that this does not change the number of connected components of the forest.
Now we apply the approximation algorithm in the claim to the instance $\mathcal{I}_B$, obtaining some solution $S_B$. 
We convert $S_B$ into a solution $S_A$ for $\mathcal{I}_A$ by mapping back each link of type $\{u_i,v_j\}$ in $S_B$ to the corresponding link $\{u,v\}$ in $\Iscr_A$. 
Observe that the solution $S_A$ might not be feasible since it could contain two copies $e_1$ and $e_2$ of the same link $\{u,v\}$.
To resolve this issue in any such case we remove $e_2$ from $S_A$ and, in case $e_1$ is a bridge of $F\cup S_A\setminus \{e_2\}$, we add an arbitrary link $\ell \neq e_1$ connecting the two connected components of $F\cup S_A\setminus \{e_1,e_2\}$. 
Such a link $\ell$ must exist since by assumption $F\cup L$ is $2$-edge-connected. 
This results in a feasible solution $S_B$ to the instance $B$ with $|S_A|\leq |S_B|$.

Let $\OPT_A$ and $\OPT_B$ be some optimal solutions to the instances $A$ and $B$, respectively.
We now show that $|\OPT_B|\leq |\OPT_A|$. 
Indeed, a solution of size $|\OPT_A|$ for $B$ can be constructed from $\OPT_A$ as follows. 
For an isolated node $v$, we replace one link $\{u,v\}\in \OPT_A$ by the link $\{u, v_1\}$ and all other links $\{u,v\}\in \OPT_A$ by the corresponding links $\{u,v_2\}$.
Suppose the resulting link set $S$ is not a feasible solution for the instance $\mathcal{I}_B$, i.e., there exists a cut that contains less than two edges in $F\cup S$.
Because $\OPT_A$ is a feasible solution for the instance $\Iscr_A$ (that can be obtained from $\Iscr_B$ by the contraction of the edge $\{v_1,v_2\}$), this cut must contain the edge $\{v_1,v_2\}$.
Now consider the unique link $\{u,v\}\in \OPT_A$ that we replaced by the link $\{u, v_1\}$.
Because $\OPT_A$ is a feasible solution for the instance $A$, the edge set $F\cup \OPT_A$ contains two edge-disjoint $u$-$v$ paths, only one of which can use the edge $\{u,v\}$.
The other path corresponds to a $u$-$v_2$ path in $F \cup \OPT_B$ that does not contain the edge $\{u,v_1\}$ and hence also not the edge $\{v_1,v_2\}$ because by construction, $v_1$ has degree two in $F\cup S$.
But then extending this $u$-$v_2$ path by the edge $\{v_1,u\}$ yields an $v_1$-$v_2$ path in $F \cup S$ that does not contain the edge $\{v_1,v_2\}$.
This contradicts the existence of a cut of size less than two that contains the edge $\{v_1,v_2\}$.
We conclude that $S$ is feasible and thus indeed $|\OPT_B|\leq |\OPT_A|$. 

Recall that the number $\nc$ of connected components of the forest is the same in the instances $\Iscr_A$ and $\Iscr_B$.
We obtain
\[
|S_A|\ \leq\ |S_B|\ \leq\ \rho \cdot |\OPT_B| + K \cdot (|\OPT_B| - \nc) \ \leq\ \rho \cdot |\OPT_A| + K \cdot (|\OPT_A| - \nc),\]
hence the claimed approximation factor.
\end{proof}

\begin{lemma}\label{lem:reducing_to_pap}
Given a polynomial time $(\rho,K)$-approximation algorithm for PAP without isolated nodes for some constants $\rho \ge 1$ and $K>0$, there is polynomial-time $(\rho, K')$-approximation algorithm for FAP without isolated nodes, where $K':=K + 2 (\rho-1)$.
\end{lemma}
\begin{proof}
Given a FAP instance $(V,F,L)$ without isolated nodes, we transform it into a PAP instance without isolated nodes by iteratively applying the following procedure. Consider any tree $T$ in the current forest with at least $3$ leaves (notice that each tree in $F$ has at least $2$ leaves by assumption). Let $e=\{v,u\}$ be any edge of $T$ incident on a node $v$ of degree at least $3$ such that the component of $T\setminus \{e\}$ containing $u$ is a path (the latter path might be an isolated node).
We delete $e$ from $F$, add a dummy node $w$ to $V$, add a dummy link $e'=\{v,w\}$ to $L$, and add a dummy edge $e''=\{w,u\}$ to $F$.

This creates a PAP instance $(V',F',L')$ without isolated nodes. 
To this PAP instance we apply a $(\rho,K)$-approximation algorithm, obtaining a solution $\mathrm{APX}'$. 
Then we return the solution $\mathrm{APX}$ for the original FAP instance obtained from $\mathrm{APX}'$ by removing all dummy links.

Let $q$ be the number of iterations of the above procedure. Notice that every dummy edge and link must belong to every feasible solution to $(V',F',L')$ since dummy nodes have degree $2$ (in $F'\cup L'$). Hence $|\mathrm{APX}|\leq |\mathrm{APX}'|-q$.
We observe that each  of the $q$ many steps of the above modification procedure increases the value of an optimal solution by at most $1$ and increases the number of connected components of the forest by $1$.
Thus,
\begin{align*}
|\mathrm{APX}|\ \le&\ |\mathrm{APX}'| -q \\
\le&\ \rho \cdot \opt' + K \cdot (\opt' - n'_{\mathrm{comp}}) -q\\
\le&\ \rho \cdot (\opt + q) + K \cdot (\opt - \nc)-q=\rho \cdot \opt + K \cdot (\opt-\nc)+(\rho-1)q,
\end{align*}
where $\opt$ and $\opt'$ denote the optimal values of the FAP and PAP instance and $\nc$ and $n'_{\mathrm{comp}}$ denote the number of connected components of $(V,F)$ and $(V',F')$, respectively.

Each tree $T$ in the original forest $(V,F)$ with $k\geq 2$ leaves contributes with $k-2$ to the total value of $q$. Hence
$$
q=n_{\mathrm{leaf}}-2\nc\leq 2\opt-2\nc,
$$
where $n_{\mathrm{leaf}}$ is the number of leaves in $(V,F)$, and the inequality follows from the fact that each leaf in $(V,F)$ must have a link of $\OPT$ incident to it, thus $\opt\geq n_{\mathrm{leaf}}/2$. The claim follows. 
\end{proof}

\end{document}